\documentclass{article}
\usepackage[margin=1in]{geometry}
\usepackage[onehalfspacing]{setspace}
\usepackage{amsmath,mathtools}
\usepackage{pgfplots}
\usepackage{subcaption}
\usepackage{booktabs}
\usepackage[braket, qm]{qcircuit}
\usepackage{graphicx}
\usepackage{authblk}
\usepackage{xcolor}
\usepackage[linesnumbered,ruled,vlined]{algorithm2e}
\SetCommentSty{mycommfont}

\SetKwInput{KwInput}{Input}                % Set the Input
\SetKwInput{KwOutput}{Output}              % set the Output
\SetKwProg{Fn}{Function}{}{}

\newcommand{\myheight}{1.79 cm}
\newcommand{\mywidth}{0.225\textwidth}
\newcommand{\QWHT}{\mathcal{H}_Q}

\usepackage[english]{babel}

\usepackage{pdflscape}
\usepackage{fancyvrb}
\usepackage{calc}
\usepackage{rotating}
\usepackage{pgf,tikz}
\usepackage{mathrsfs}
\usetikzlibrary{arrows}
\usepackage{mathtools}

\usepackage{listings}
\usepackage[]{qcircuit}
\usepackage{braket}
\usepackage{mathtools}
\usepackage{varwidth}
\usepackage{caption}
\usepackage{subcaption}
\usepackage{graphicx}
\usepackage[export]{adjustbox}

\newcommand{\norm}[1]{\left\lVert#1\right\rVert}

\makeatletter
\def\paragraph{\@startsection{paragraph}{4}%
	\z@\z@{-\fontdimen2\font}%
	{\normalfont\bfseries}}
\makeatother

\usepackage{graphicx}
\usepackage{pgffor}
\usepackage{tikz,tikz-cd}
\usetikzlibrary{backgrounds,fit,decorations.pathreplacing}  % TikZ libraries
\usepackage{booktabs}
\usepackage{enumerate}

\usepackage{amssymb, amsmath, amsthm}

\usepackage{xypic}% % % % % % % % % % % % % % % % % % % % % %
\usepackage{scalerel}%%%%%%%%%%%%%%%%%%%%%%%%%%%%%%%%%%%%%%%%%%%%%%%%%%%%%%%%%%%%%%%%%%%%%%%%%%%%
\usepackage{ulem}

\usepackage{graphicx}              % to include figures
\usepackage{amsmath}               % great math stuff
\usepackage{amsfonts}              % for blackboard bold, etc
\usepackage{amsthm}                % better theorem environments
\usepackage{xkeyval}               % for passing more than 9 arguments to newcommand
\usepackage{amssymb}
\usepackage{enumerate}             %for lisitng items with roman (i), (ii) etc.
\usepackage{color}
\usepackage{breqn}                 % for multiline \left( \right)
\usepackage{bm}                    % for bold mathematical symbol 

\allowdisplaybreaks[1]
\usepackage{nicefrac}
\usepackage{booktabs}

\usepackage{kpfonts}
\usepackage{stackengine}
\usepackage{calc}
\newlength\shlength
\newcommand\xshlongvec[2][0]{\setlength\shlength{#1pt}%
	\stackengine{-5.6pt}{$#2$}{\smash{$\kern\shlength%
			\stackengine{7.55pt}{$\mathchar"017E$}%
			{\rule{\widthof{$#2$}}{.57pt}\kern.4pt}{O}{r}{F}{F}{L}\kern-\shlength$}}%
	{O}{c}{F}{T}{S}}

\usepackage{hyperref}
\hypersetup{
	colorlinks   = true,    % Colours links instead of ugly boxes
	citecolor    = red      % Colour of citations
}

\newcommand{\RN}[1]{%
	\textup{\uppercase\expandafter{\romannumeral#1}}%
}

\allowdisplaybreaks

\newcommand{\meqref}[1]{\text{Eq}.~\eqref{#1}}

\newtheorem{thm}{Theorem}[subsection]

\newtheorem{lem}[thm]{Lemma}

\newtheorem{defn}[thm]{Definition} 
\newtheorem{example}[thm]{Example} 
\newtheorem{remark}[thm]{Remark}

\newcommand{\RR}{\mathbb{R}}      % for Real numbers
\newcommand{\ZZ}{\mathbb{Z}}      % for Integers

\newcommand{\mat}[4]{\left[\begin{smallmatrix*}[r]
		#1 & #2 \\
		#3 & #4 \\
	\end{smallmatrix*}\right]}

\arraycolsep=3pt\def\arraystretch{0.5}

\def\CC{{\mathbb C}}

\def\RR{{\mathbb R}}

\def\ZZ{{\mathbb Z}}

\def\<{\langle}
\def\>{\rangle}

\def\longto{\longrightarrow}

\DeclareMathOperator{\GL}{GL}

\newcommand{\iprod}[2]{\langle #1, \, #2 \rangle}

\numberwithin{equation}{section}

\begin{document}

	\title{A hybrid classical-quantum algorithm for solution of nonlinear ordinary differential equations%\thanks{Grants or other notes
	}

	\author[1]{Alok Shukla \thanks{alok.shukla@ahduni.edu.in}}
	\author[2]{Prakash Vedula \thanks{pvedula@ou.edu}}
	\affil[1]{Ahmedabad University, India}
	\affil[2]{University of Oklahoma, USA}

	\maketitle

	\begin{abstract}
		A hybrid classical-quantum approach for the solution of nonlinear ordinary differential equations using Walsh-Hadamard basis functions is proposed.
		Central to this hybrid approach is the computation of the Walsh-Hadamard transform of arbitrary vectors, which is enabled in our framework using quantum Hadamard gates along with state preparation, shifting, scaling, and measurement operations.
		It is estimated that the proposed hybrid classical-quantum approach for the Walsh-Hadamard transform of an input vector of size $ N $ results in a considerably lower computational complexity ($\mathcal{O}(N)$ operations) compared to the Fast Walsh-Hadamard transform ($\mathcal{O}(N \log_2 (N)) $ operations). 	
		This benefit will also be relevant in the context of the proposed hybrid classical-quantum approach for the solution of nonlinear differential equations.
		Comparisons of results corresponding to the proposed hybrid classical-quantum approach and a purely classical approach for the solution of nonlinear differential equations (for cases involving one and two dependent variables) were found to be satisfactory.
		Some new perspectives relevant to the natural ordering of Walsh functions (in the context of both classical and hybrid approaches for the solution of nonlinear differential equations) and the representation theory of finite groups are also presented here.
	\end{abstract}

	%\tableofcontents

	\section{Introduction}\label{sec:intro}
	
	Ordinary differential equations (ODEs) are often used for mathematical descriptions of models in many engineering and science (including those relevant to physical, biological, and social sciences) disciplines. As exact solutions to nonlinear ODEs are generally not available, various numerical methods (e.g. single-step, multi-step, finite difference, finite element, and spectral methods) are often used to obtain approximate solutions~\cite{butcher2016numerical,suli2010numerical}. The selection of appropriate numerical methods for the solution of nonlinear differential equations is guided by considerations of accuracy, numerical stability, stiffness, memory constraints and execution time.  
	
	Spectral or pseudospectral methods are often preferred for high-order accuracy~\cite{trefethen2000spectral,boyd2001chebyshev}. 
	In these methods, the solution is expressed as a linear combination of certain (global) basis functions and the coefficients of these basis functions are optimally selected based on the underlying differential equation.
	These methods can be shown to have exponential convergence. They are often more accurate compared to finite difference or finite element methods.
	
	While many different choices for orthogonal basis functions can be considered, we consider Walsh functions~\cite{walsh1923closed, beauchamp1975walsh} as basis functions for the solution of ordinary differential equations.
	Walsh functions and associated transforms have been found to be useful in many contexts relevant to signal processing~\cite{zarowski1985spectral}, image compression~\cite{kuklinski1983fast}, cryptography~\cite{lu2016walsh}, solution of non-linear differential equations~\cite{beer1981walsh,ahner1988walsh}, solution of variational problems \cite{chen1975walsh} and solution of partial differential equations relevant to fluid dynamics (with discontinuities or shocks)~\cite{gnoffo2014global, gnoffo2015unsteady, gnoffo2017solutions}.
	Our choice of basis functions is motivated not only due to important properties of Walsh functions but also due to the close natural connection between the Walsh-Hadamard transform~\cite{beauchamp1975walsh} and the Hadamard gate~\cite{nielsen2002quantum} that is widely used in many quantum circuits and algorithms. Although the Hadamard gate in a quantum circuit performs a Walsh-Hadamard transform of a quantum state (vector), it is important to note that extracting useful classical information is challenging. In this paper we address the underlying challenges associated with extracting useful classical information related to the Walsh-Hadamard transform of an arbitrary vector for subsequent use in solution of nonlinear ordinary differential equations.
	
	Walsh functions, which are square-wave like functions and have values of $\pm 1$ at any point in the domain of interest, form a complete set of orthogonal basis functions.
	Unlike many other orthogonal basis functions, Walsh functions are closed under multiplication (e.g. product of two Walsh functions results in a function that is contained in the basis set). 
	As Walsh basis functions are composed of square-wave like functions, they are well suited to represent discontinuous functions and/or rapidly changing functions.
	%\item 
	Coefficients of the Walsh basis functions can be obtained from a data vector in the physical domain using a discrete Walsh-Hadamard transform (matrix). It can be shown~\cite{kunz1979equivalence} that the discrete Walsh-Hadamard transform when applied to a $2^n$ one-dimensional data is equivalent to the discrete $n-$dimensional Fourier transform applied the same data arranged on a binary $n-$cube.
	%    
	%\item
	In contrast to many other transforms, the discrete Walsh-Hadamard transform can be computed by addition and subtraction operations alone and without the need for multiplication. This feature enables fast and accurate computation of the discrete Walsh-Hadamard transform. Note that the computational complexities associated with basic implementations of addition and multiplication operations of two $p-$digit numbers can be estimated to be $\mathcal{O}(p)$ and $\mathcal{O}(p^2)$ respectively.
	While a basic implementation of the discrete Walsh-Hadamard transform has a computational complexity of $\mathcal{O}(N^2)$ (where $N$ denotes the elements in the data vector), efficient algorithms similar to those used in the Fast Fourier Transform (FFT) exist for computation of Fast Walsh-Hadamard transform that has a complexity of $\mathcal{O}(N ~\log N)$.
	
	Walsh functions can be ordered in many ways and commonly used orderings include the sequency ordering (according to the number of zero-crossings), the natural ordering (based on Hadamard matrices generated by Kronecker products) and the dyadic ordering~\cite{geadah1977natural}. While the matrix rows for the transform matrices corresponding to these orderings can be interchanged via the use of permutation matrices, more efficient approaches exist for such purposes. For instance, conversion from the natural to the sequency ordering can be accomplished via a combination of bit reversal and Gray code to binary conversion.
	
	%\item 
	In the literature, many treatments of Walsh basis functions applied to the solution of differential equations have focused on the sequency ordering. In contrast, in this work, we consider the natural ordering of Walsh functions and present results (including integration matrices and associated inverses) based on this ordering.
	For convenience, we will refer to Walsh basis functions considered in the natural ordering (based on Hadamard matrices) as Walsh-Hadamard (WH) basis functions.
	
	%\item 
	Our preference for this ordering is based on the insight that the transform matrix for this ordering corresponds to the Hadamard matrix describing the quantum Hadamard gate commonly used in many quantum circuits and algorithms.
	%  
	%\item
	This insight provides further motivation to explore the possibility of making use of quantum Hadamard gates for fast computation of the Walsh-Hadamard transform which is a key step in the solution of differential equations based on Walsh functions.
	%    
	%\item
	Quantum algorithms were previously proposed for the solution of linear ordinary differential equations~\cite{childs2020quantum,berry2017quantum, berry2014high} and nonlinear ordinary differential equations~\cite{leyton2008quantum,lloyd2020quantum,liu2021efficient}.
	%
	%\item 
	It is worth noting that development of quantum algorithms in the latter category is particularly challenging owing to the linearity of quantum mechanics.
	%
	%\item
	Limitations of previous quantum algorithms for the solution of nonlinear differential equations can be associated with poor scaling (as resources needed increase exponentially~\cite{leyton2008quantum} or quadratically~\cite{lloyd2020quantum} with integration time), restricted sparsity structure~\cite{lloyd2020quantum}, dissipation, restricted nonlinearities or approximation errors associated with Carleman linearization~\cite{liu2021efficient}, to name a few.
	In this paper, we propose a hybrid classical-quantum approach for the solution of nonlinear ordinary differential equations (ODEs) using Walsh-Hadamard basis functions. 
	We present new perspectives for the solution of nonlinear ODEs for both classical and hybrid approaches based on the natural ordering of Walsh functions. 
	In this context, we also explore some theoretical insights on Walsh-Hadamard transform arising from the character theory of finite groups.
	While Hadamard gates are commonly used in many quantum circuits and algorithms and can be used to naturally obtain the Hadamard transform (under certain conditions), there are challenges involved in extracting useful classical information particularly due to the ambiguity associated with the global phase while carrying out measurements. We address these challenges via a novel hybrid classical-quantum approach for computation of Walsh-Hadamard transform for arbitrary input vectors. This approach makes use of the special structure of the Walsh-Hadamard transform and it involves a combination of shifting, scaling and measurement operations, along with state preparation and use of quantum Hadamard gates. The advantage of our hybrid approach is a considerably lower computational complexity ($\mathcal{O}(N)$) in comparison to classical Fast Walsh-Hadamard transform ($\mathcal{O}(N \log_2 N)$). This speedup can also be harnessed for the solution of nonlinear ODEs. The proposed hybrid classical-quantum approaches for computation of Walsh-Hadamard transform (using quantum Hadamard gates) and solutions of nonlinear ordinary differential equations were successfully implemented and tested using the simulated environment of Qiskit (IBM's open source quantum computing platform).  

	While we demonstrate our proposed approach on solutions of nonlinear Initial Value Problems (IVPs), the proposed approach could also be easily extended to solution of nonlinear Boundary Value Problems (BVPs) as well, via the use of a shooting method (\cite{ascher1995numerical}, \cite{bulirsch2002introduction}, \cite{suli2003introduction}, \cite{burden2015numerical},  and \cite{isaacson2012analysis}).

	The rest of this paper is organized as follows. 
	In section~\ref{sec:wh_functions}, we present the definition, construction and properties of Walsh-Hadamard basis functions. The relation between Walsh functions and the character theory of finite groups is presented in section~\ref{Section_Walsh_character_theory}. Classical and quantum implementation of the Walsh-Hadamard transform are discussed in sections~\ref{subsec:wh_transform} and \ref{subsec:quantum_wh_transform}. Our proposal for a hybrid classical-quantum approach for Walsh-Hadamard transform is presented in section~\ref{sec:hybrid-cq-wht} and the associated computational complexity is discussed in section~\ref{sec:comp_complexity}. This approach is central to the solution of nonlinear ordinary differential equations using Walsh-Hadamard basis functions. Some basic concepts relevant to function representation, integration and differentiation using Walsh-Hadamard basis functions are presented in sections~\ref{sec:representation_wh_basis}, \ref{subSec_Integraiton_matrix} and \ref{subsec:differentiation_matrix}. The proposed hybrid quantum-classical approach for the solution of (non-stiff and stiff) nonlinear ordinary differential equations is presented in section~\ref{sec:hybrid_algorithm_de}. Three computational examples based on this approach are presented in section~\ref{sec:computational_examples} and conclusions are summarized in section~\ref{sec:conclusion}.

	\section{Walsh Functions}\label{sec:wh_functions}
	Let $ N=2^n $ be a positive integer. For $ j =0,~1,~2, ~\ldots~ N-1 $, the Walsh functions are defined as follows
	
	\begin{align}
		W_0(x) &= 1 \quad \text{for } 0 \leq x \leq 1,  \\
		W_{2j} (x) &= W_j(2x) + (-1)^j W_j (2x -1 ),  \\
		W_{2j+1} (x) &= W_j(2x) - (-1)^j W_j (2x -1 ), \\
		W_j(x) &= 0 \quad \text{for } x < 0 \text{ and } x >1.
	\end{align}
	For $ N =8 $ the Walsh functions are shown in Figure $\ref{fig_walsh_sequency}$. It is clear that the Walsh functions take the values $ 1 $ or $ -1 $ on the interval $ [0,1] $ and they are $ 0 $ everywhere else. The Walsh functions shown in Figure $\ref{fig_walsh_sequency}$ are in the so-called \textit{sequency order}. In sequency ordering the number of sign changes (or zero-crossings) for the Walsh functions increase as the orders of the functions increase. The Walsh functions with an even number of sign changes are symmetric about the line $ x = 1/2 $ and the Walsh functions with odd numbers of sign changes are anti-symmetric about the line $ x = 1/2 $. In the literature, analogous to the trigonometric sine and cosine functions the $ sal_j(x) $ and $ cal_j(x) $ functions are defined as given below.
	\begin{align}
		sal_j(x) &= W_{2j-1} (x), \\
		cal_j(x) & = W_{2j} (x).
	\end{align} 
	It can be shown that Walsh functions are orthonormal.
	\begin{equation}\label{eq_orthonormal}
		\int_{- \infty}^{\infty} \,  W_m(x) \, W_n(x) \, dx = \begin{cases}
			&N, \quad \text{if } m =n, \\
			&0, \quad \text{if } m \neq n. 
		\end{cases}
	\end{equation}
	{
		\begin{figure}[htp!]
			\begin{center}
				\begin{subfigure}[t]{\mywidth}
					\centering
					\includegraphics[width=\linewidth,height=\myheight]{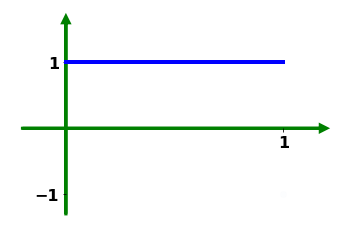}
					\caption{$ W_0(x) $}
				\end{subfigure}
				%	\hfill
				\hspace{1cm}
				\begin{subfigure}[t]{\mywidth}
					\centering
					\includegraphics[width=\linewidth,height=\myheight]{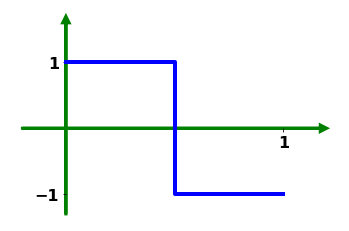}
					\caption{$ W_1(x) $}
				\end{subfigure}
				
				\medskip
				
				\begin{subfigure}[t]{\mywidth}
					\centering
					\includegraphics[width=\linewidth,height=\myheight]{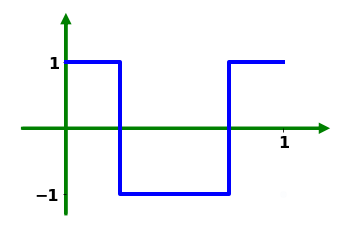}
					\caption{$ W_2(x) $}
				\end{subfigure}
				%	\hfill
				\hspace{1cm}
				\begin{subfigure}[t]{\mywidth}
					\centering
					\includegraphics[width=\linewidth,height=\myheight]{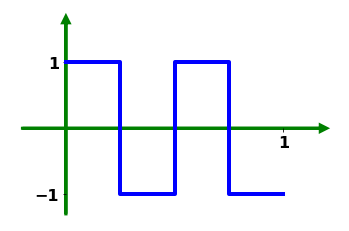}
					\caption{$ W_3(x) $}
				\end{subfigure}

				\medskip
				
				\begin{subfigure}[t]{\mywidth}
					\centering
					\includegraphics[width=\linewidth,height=\myheight]{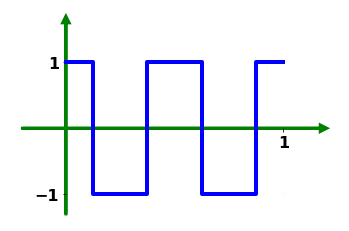}
					\caption{$ W_4(x) $}
				\end{subfigure}
				%	\hfill
				\hspace{1cm}
				\begin{subfigure}[t]{\mywidth}
					\centering
					\includegraphics[width=\linewidth,height=\myheight]{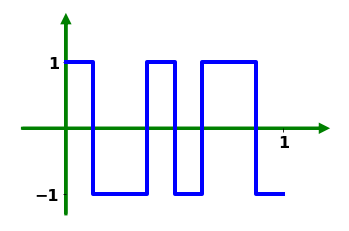}
					\caption{$ W_5(x) $}
				\end{subfigure}

				\medskip
				
				\begin{subfigure}[t]{\mywidth}
					\centering
					\includegraphics[width=\linewidth,height=\myheight]{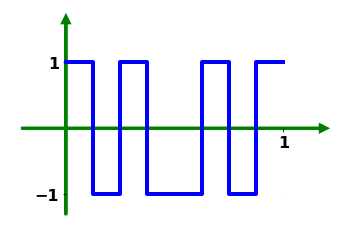}
					\caption{$ W_6(x) $}
				\end{subfigure}
				%	\hfill
				\hspace{1cm}
				\begin{subfigure}[t]{\mywidth}
					\centering
					\includegraphics[width=\linewidth,height=\myheight]{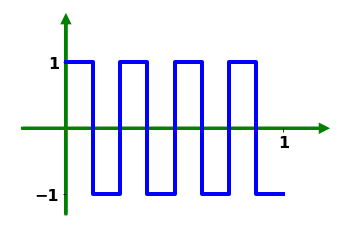}
					\caption{$ W_7(x) $}
				\end{subfigure} 
				
			\end{center} 
			
			\caption{Walsh functions in the sequency ordering for $ N=8 $.}\label{fig_walsh_sequency}
		\end{figure}
		
		We refer the readers to~\cite{beauchamp1975walsh} for alternate definitions and properties of Walsh functions. Walsh functions can also be interpreted using the character theory of finite groups. Perhaps, this provides a conceptually cleaner approach to defining an equivalent collection of Walsh functions up to ordering. We will define Walsh functions using this approach in Section~\ref{Section_Walsh_character_theory}.
		In the rest of the paper, we will not use the sequency ordering of Walsh functions, instead, we will use the natural Hadamard ordering as used in the quantum computing community. Of course, it is easy to go from one ordering to another in applications.
		
		\section{Walsh Functions and the character theory of finite groups} \label{Section_Walsh_character_theory}
		We will first discuss some well-known results related to the Fourier analysis on finite abelian groups. We will work with a general finite abelian group in the beginning and later specialize to $ (\ZZ/2\ZZ)^n $, leading to the theory of Walsh-Hadamard Transform. All the results in this section can be found in any standard textbook covering the representation theory of finite groups (e.g.,~\cite{serre1977linear},~\cite{fulton2013representation},~\cite{steinberg2012representation},~\cite{terras1999fourier}). While the character theory of finite groups is well-known its connections with Walsh-Hadamard transform are not adequately explored in the literature (to the best of our knowledge). Here we will explore these connections further and we will show that Walsh functions can be interpreted in terms of characters of certain abelian groups.
		
		\subsection{Fourier Analysis on Finite Groups}
		Let $ G $ be a finite group. We denote the vector space of all complex valued functions on $ G $ by $ L(G) $, i.e.,
		\begin{equation}\label{eq_def_LG}
			L(G) = \{ \, f \ | \ f: G \longto \CC\}.
		\end{equation}
		The addition and scalar multiplication in $ L(G) $ are given by
		\begin{equation}\label{eq_def_vector_space_LG}
			(f_1 + f_2) (g) = f_1(g) + f_2(g) \quad \text{ and } \quad (cf)(g) = c f(g), \quad \text{for $ f_1, f_2 \in L(G) $, and $ c \in \CC $. }
		\end{equation} 
		The space $ L(G) $ is an inner product space with the inner product 
		\begin{equation}\label{eq_def_inner_product}
			\iprod{f_1}{f_2} = \frac{1}{|G|} \, \sum_{g \in G} \, f_1(g) \overline{f_2 (g)}.
		\end{equation}
		Next, we recall some basic facts related to the representation theory of finite groups. Let $ \phi : G \longto \GL(V) $  be a representation of $ G $, where $ V $ is a finite-dimensional vector space. It means each element $ g \in G $ gives rise to a linear transformation $ \phi(g) $ of the vector space $ V $. The representation  $ \phi $ of the group $ G $ can be thought of as the action of $ G $  on the vector space $ V $, such that $ \phi(g h ) \cdot  v = \phi (g) \cdot (\phi(h) \cdot v)$ for all $ g,h \in G $ and $ v \in V $, and $ e \cdot v =v  $ for the identity $ e $ of $ G $.  The dimension of $ V $ is called the dimension of $ \phi $. A subspace $ U $ of $ V $ is called $ G $-invariant if $ \phi(g) \cdot u \in U $ for all $ u \in U $ and $ g \in G $. A non-zero representation $ \phi  $ is said to be \textit{irreducible} if $ \{0\} $ and $ V $ are the only $ G $-invariant subspaces of $ V $. Two $ n $-dimensional representations $ \phi  $ and $ \psi $ are called \textit{equivalent representations} if there exists an invertible matrix $ T \in \GL_n(\CC) $ such that $ \phi(g) = T \, \psi(g) \, T^{-1} $ for every $ g \in G$. If $ \phi $ and $ \psi $ are equivalent representations then we write $ \phi \sim \psi $. A finite-dimensional representation $ \phi $ is called unitary if
		$ \iprod{\phi(g) \cdot u}{\phi(g) \cdot v} = \iprod{u}{v} $ for all $ u,v \in V $ and $ g \in G $. It is a fact that every representation of a finite group $ G $ is equivalent to a unitary representation. It is easy to see that a one dimensional unitary representation is a homomorphism $ \phi : G \longto S^1 $, where $ S^1 = \{ z \in \CC \ | \ |z| =1\} $.
		
		The \textit{character} of the representation $ \phi $ is a function $ \chi : G \longto \CC $  defined as  $ \chi_{\phi} (g)  = Tr(\phi (g)) $, where $ Tr $ denotes the trace of a matrix. It follows from the definition of the character $ \chi_{\phi}  $ that $ \chi_{\phi}(g x g^{-1}) = \chi_{\phi}(x) $ for all $ g,x \in G $, and it means that $ \chi_{\phi}  $ is a \textit{class function}, i.e., it is constant on conjugacy classes of $ G $. A character of an irreducible representation is called an irreducible character. 
		
		The irreducible characters of a group, as their name suggests, characterize inequivalent irreducible representations of a group. The following orthogonality result is of fundamental importance in the representation theory of finite groups.
		
		\begin{thm}\label{thm_orthogonality_one}
			Let $ \phi $ and $ \psi $ be irreducible representations of $ G $. Then
			\begin{equation}\label{eq_orthogonality}
				\iprod{\chi_{\phi}}{\chi_{\psi}}  = \begin{cases}
					1, &\quad \text{if } \phi \sim \psi, \\
					0, &\quad \text{otherwise}.
				\end{cases} 
			\end{equation}
			It follows that the irreducible characters of $ G $ form an orthonormal set of class functions. The number of inequivalent irreducible representations is the same as the number of inequivalent irreducible characters of the group $ G $, which in turn is given by the number of distinct conjugacy classes of $ G $.
		\end{thm}

		Now onward \textit{it is assumed that $ G $ is a finite abelian group}. 
		%It is a consequence of Schur's lemma that the irreducible representations of any abelian group $ G $ is one dimensional. 
		Since $ G $ is assumed to be abelian, there are $ |G| $ number of conjugacy classes of $ G $, and so there are $ |G| $ number of inequivalent irreducible characters of $ G $. The set of these characters form a group of order $ |G| $, known as the dual group of $ G $. More precisely, the \textit{dual group} of $G$, denoted by $ \widehat{G} $, is defined as 
		\begin{equation}\label{eq_def_dual_group}
			\widehat{G} = \{ \chi \ | \  \chi  :  G \longto S^{1} \text{ is an irreducible character of $ G $}\}.
		\end{equation} 
		The group operation is pointwise multiplication, i.e.,  $ (\chi_{1} \cdot \chi_2) (g) =  \chi_{1}(g) \chi_2(g)$ for all $ g\in G $ and $ \chi_{1},\chi_2 \in \widehat{G} $. 
		%The identity of $ \widehat{G} $ is the trivial character $ \chi_{id} $ such that $ \chi_{id} (g) =1 $ for all $ g \in G $. 
		
		We have the following orthogonality relation. 
		\begin{thm}\label{thm_orthogonality_two}
			Let $ \chi_1, \chi_2 \in \widehat{G} $. Then
			\begin{equation}\label{eq_orthogonality_two}
				\iprod{\chi_{1}}{\chi_{2}}  = \begin{cases}
					1, &\quad \text{if } \chi_1 = \chi_2, \\
					0, &\quad \text{otherwise}.
				\end{cases} 
			\end{equation}
		\end{thm}
		The irreducible characters of $ G $ form a basis of $ L(G) $ over $ \CC $, therefore any $ f \in L(G) $ can be written as 
		\begin{equation}\label{eq_decomposition_f}
			f(g) = \sum_{\chi_i \in \widehat{G}} \, c_{i} \, \chi_i (g), \quad \text{where } c_{i} \in \CC  \text{ for all } \chi_i \in \widehat{G}.
		\end{equation}
		It follows from Theorem \ref{thm_orthogonality_two} that $ 	c_i = \iprod{f}{\chi_i} $. Therefore,
		\begin{equation}\label{eq_ci}
			f(g) = \sum_{\chi_i \in \widehat{G}} \, \iprod{f}{\chi_i} \, \chi_i (g) = \frac{1}{\sqrt{|G|}} \sum_{\chi \in \widehat{G}} \, \widehat{f} (\chi) \chi(g),
		\end{equation}
		where the Fourier transform $ \widehat{f} \ : \ \widehat{G} \longto \CC $ is defined below.
		\begin{defn} \label{def_fourier_transform_def}
			Let $ G  $ be a finite abelian group and $ f \in L(G) $. Then the Fourier transform $ \widehat{f} \ : \ \widehat{G} \longto \CC $ is defined by
			\begin{equation}\label{def_Fourier_transform}
				\widehat{f} (\chi) = \sqrt{|G|} \, \iprod{f}{\chi} = \frac{1}{\sqrt{|G|} } \, \sum_{g \in G} \, f(g) \overline{\chi (g)}.
			\end{equation}
		\end{defn}
		The \meqref{def_Fourier_transform} and \meqref{eq_ci} define the direct and inverse Fourier transform respectively. It can be verified that the map $ f \longto \widehat{f}  $ defines a linear transformation from $ L(G) $ to $ L(\widehat{G}) $.
		
		\subsection{Walsh Functions as Characters of the Group $ (\ZZ/2\ZZ)^n $}
		
		Let $ \omega_n = e^{\frac{2 \pi i }{n}} $. Consider $ G = \ZZ/n\ZZ   $. Then $ \widehat{G} = \{\chi_0, \chi_{1}, ~\ldots~, \chi_{n-1}  \} $, where $ \chi_{k} (m) = \omega_n^{km} $ for
		$ k=0,~1,~\ldots~, n-1 $. The well-known Discrete Fourier Transform is the Fourier transform for the complex valued functions defined on the group $ G = \ZZ/n\ZZ  $.
		
		\begin{example} \label{ex_character}
			For $ G = \ZZ/2\ZZ = \{[0],[1]\} $,  the dual group $ \widehat{G} $ consists of  $ \{\chi_0, \chi_1\} $ such that 
			\[
			\chi_0([0]) = 1, ~\chi_0([1]) = 1, \qquad \chi_1([0]) = 1, ~\chi_1([1]) = -1.
			\]
			Here $ [0] $ and $ [1] $ denote the equivalence classes of $ 0 $ and $ 1 $, respectively. For easing the notation, in the following, we will just write $ 0 $ and $ 1 $ for $ [0] $ and $ [1] $, respectively. The character table of the group $ G  \ZZ/2\ZZ  = \{0,1\}  $ is shown in Table~\ref{tab_one}. 
			\begin{table}[ht]
				\centering
				\begin{tabular}{rrr}
					\hline \\
					& $ 0 $ & $ 1 $   \\
					\hline \\
					$ \chi_{0} $	& $ 1 $ & $ 1 $\\ \\
					$ \chi_{1} $	& $ 1 $ & $- 1 $   \\ \\
					\hline \\
				\end{tabular}
				\caption{The character table of $ G = \ZZ/2\ZZ  = \{0,1\}  $.}\label{tab_one}
			\end{table}
			We recall the Hadamard gate $ H $, which sends the qubit $ \ket{0} $ to $ \frac{1}{\sqrt{2}} \left( \ket{0} + \ket{1}\right) $ and the qubit $ \ket{1} $ to $ \frac{1}{\sqrt{2}} \left( \ket{0} - \ket{1}\right) $. It is interesting to note the similarity of the character table of $ G = \ZZ/2\ZZ $ with the transformation matrix of the Hadamard gate $ H $ given by
			\[
			H \equiv \, \frac{1}{\sqrt{2}} \,  \mat{1}{1}{1}{-1}.
			\]   
		\end{example}
		%One can easily check that $ \widehat{G_1 \times G_2}  = \widehat{G_1} \times \widehat{G_2}$.
		Let $ \chi_i \in \widehat{G_1} $ and $ \chi_j \in \widehat{G_2} $. Then $ \chi_i \otimes \chi_j \in  \widehat{G_1 \times G_2} $, such that $ (\chi_i \otimes \chi_j) (g_1, g_2) = \chi_i(g_1) \chi_j(g_2) $. It can be shown that if $ \widehat{G_1} = \{\chi_{i} \ | \ i=1,~\ldots~, m\} $ and   $ \widehat{G_2} = \{\chi_{j} \ | \ j=1,~\ldots~, n\} $, then 
		$ \widehat{G_1 \times G_2} = \{ \chi_{(ij)} := \chi_{i} \otimes \chi_{j} \ | \  i=1,~\ldots~, m, j=1,~\ldots~, n \} $.
		
		Note that, in the remainder of the paper, the string of binary digits will be represented without parentheses to ease the notation (where it is clear from the context). 
		For example, depending upon the context, $ \chi_{(00)}$ and $ \chi_{(01)} $ will be denoted as $ \chi_{00} $ and $ \chi_{01} $, respectively. 
		\begin{example}
			For $ G = \ZZ/2\ZZ \times \ZZ/2\ZZ  = \{0,1\} \times \{0,1\} = \{ 00, 01, 10, 11\} $,  the dual group $ \widehat{G}  = \{ \chi_{(ij)} = \chi_i \otimes \chi_j \ | \ i=0,1, j=0,1 \} $. 
			
			Again we note the similarity between the  character table of $ G = \ZZ/2\ZZ \times \ZZ/2\ZZ $ (see Table~$ \ref{tab_two} $) and the transformation matrix of \[ H \otimes H \equiv \, \frac{1}{2} \, \begin{pmatrix*}[r]
				$ 1 $ & $ 1 $  & $ 1 $  & $ 1 $  \\
				$ 1 $ & $ -1 $ & $ 1 $  & $ -1 $ \\ 
				$ 1 $ & $ 1 $ & $ -1 $ & $ -1 $ \\ 
				$ 1 $ & $ -1 $ & $ -1 $  & $ 1 $ \\ 
			\end{pmatrix*} .\]  
			
			\begin{table}[ht]
				\centering
				\begin{tabular}{rrrrr}
					\hline \\
					& $ 00 $ & $ 01 $ & $ 10 $ & $ 11 $  \\
					\hline \\
					$ \chi_{00} $	& $ 1 $ & $ 1 $  & $ 1 $  & $ 1 $  \\ \\
					$ \chi_{01} $	& $ 1 $ & $ -1 $ & $ 1 $  & $ -1 $ \\ \\
					$ \chi_{10} $	& $ 1 $ & $ 1 $ & $ -1 $ & $ -1 $ \\ \\
					$ \chi_{11} $	&  $ 1 $ & $ -1 $ & $ -1 $  & $ 1 $ \\ \\
					\hline \\
				\end{tabular}
				\caption{The character table of $ G = \ZZ/2\ZZ \times \ZZ/2\ZZ  = \{0,1\} \times \{0,1\} = \{ 00, 01, 10, 11\} $. Note that $ \chi_{ij}  $ denotes  $ \chi_{(ij)}  $ in the above table.}\label{tab_two}
			\end{table}  
		\end{example}
		The above examples can easily be generalized to the group $ G = \{ 0,1\}^n $, which is a group of order $ 2^n $ with elements identified with all the binary words of length $ n $.  The elements of this group can also be written as integers $ 0,~1,~2,~\ldots~, ~2^{n}-1 $, where the integer $ k $ corresponds to the string $ (k_{n-1}k_{n-2}~\ldots~ k_{0}) $ which is the binary representation of the integer $ k $, i.e., $ k = \sum_{i=0}^{n-1} \, k_i 2^i  $ with $ k_i \in \{0,1\} $. For example, for $ n=2 $, the elements of $ G = \{ 0,1\}^2 $ can be written as  $ \{0,~1,~2,~3\}  $, where $ 0 $ is identified with $ (00) $, $ 1 $ is with $ (01) $, $ 2 $ is with $ (10) $ and $ 3 $ is with $ (11) $.  A similar convention is followed for the character $ \chi_k $, i.e., 
		$ \chi_k = \chi_{k_{n-1}} \otimes \chi_{k_{n-2}} \otimes ~~\ldots~~ \otimes  \chi_{k_{1}} \otimes \chi_{k_{0}}  = \chi_{(k_{n-1}k_{n-2} ~\ldots~,  k_1 k_{0})}$. For example, $ \chi_{2} $ denotes $ \chi_{(10)} = \chi_1 \otimes \chi_0 $, and $ \chi_{3}$ denotes $ \chi_{(11)} = \chi_{1} \otimes \chi_{1} $. It is clear that using this convention for the group $ G = ( \ZZ/2\ZZ )^n  = \{0,1\}^n$  the dual group is $ \widehat{G} = \{ \chi_{k} \ | \  k=0,~1,~\ldots~, 2^n -1 \} $. More explicitly, suppose $ k $ has the binary representation  $  (k_{n-1} k_{n-2}k_1 k_{0}) $  and  $ x \in G $ such that it has the binary representation $  (x_{n-1} x_{n-2}~\ldots~ x_1 x_{0}) $. Then
		\begin{align}
			\chi_k (x) &= \left(\chi_{k_{n-1}} \otimes \chi_{k_{n-2}} \otimes ~\ldots~ \chi_{k_{1}} \otimes \chi_{k_{0}}\right)(x) \nonumber \\
			& = \chi_{k_{n-1}}(x_{n-1}) \chi_{k_{n-2}}(x_{n-2})  ~\ldots~ \chi_{k_{1}}(x_1)  \chi_{k_{0}} (x_0) \nonumber \\
			& = (-1)^{\left(\sum_{i = 0}^{n-1} \, x_i k_i \right)} \label{eq_explicit_def_chik}
		\end{align}

		\begin{defn}
			Let $ N = 2^n $ be a positive integer and let $ G = (\ZZ/2\ZZ)^n = \{0,1\}^n  $. Assume $ 0 \leq k \leq N$ to be an integer. 
			Then the character $ \chi_{k} \in \widehat{G}  $ can be identified with the $ k^{\text{th}} $ Walsh function of order $ N $. 
		\end{defn}
		
		\begin{defn}
			Let $ N = 2^n $ be a positive integer. Suppose the interval $ [0,1] $ is divided into $ N $ equal sub-intervals each of length $ \frac{1}{N} $. These sub-intervals are ordered from left to right by assigning binary equivalent of numbers from $ 0 $, $ 1 $, $ 2 $, $ ~\ldots~ $, $ N-1 $, with $ 0 $ assigned to the left most interval. It is clear that any $ x \in [0,1] $ will lie in one of these intervals. Then the $ k^{\text{th}} $ Walsh function  of order $ N $ is the function $ W_k(x) : [0,1] \longto \{1,-1\} $   defined by 
			\begin{equation}\label{eq_def_walsh}
				W_k(x) := \begin{cases}
					&\chi_{k}(i), \qquad \frac{i}{N} \leq x < \frac{i+1}{N}, \,  0 \leq i \leq N-1,\\
					&\chi_{k}(N-1), \qquad $ x = 1 $.
				\end{cases} 
			\end{equation}
		\end{defn}

		\begin{example}
			For $ N =2^2 $ there are $ 4 $ Walsh functions as shown in Figure \ref{fig_walsh_four}. It is instructive to note the correspondence between the rows of the character table of $ G = \ZZ/2\ZZ \times \ZZ/2\ZZ  $ in Table \ref{tab_two} and the graph of  Walsh functions shown in Figure \ref{fig_walsh_four} below.
			
			\begin{figure}[ht]
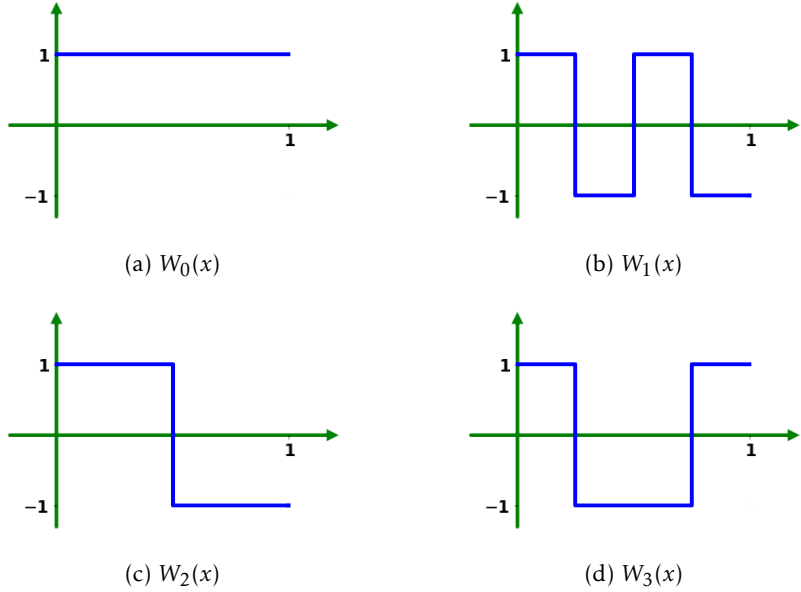

				
				%	\resizebox{.9\linewidth}{!}{
					
					\begin{center}
						\begin{subfigure}[t]{.3\textwidth}
							\centering
							\includegraphics[width=\linewidth]{W0.png}
							\caption{$ W_0(x) $}
						\end{subfigure}
						%	\hfill
						\hspace{1cm}
						\begin{subfigure}[t]{.3\textwidth}
							\centering
							\includegraphics[width=\linewidth]{W1.png}
							\caption{$ W_1(x) $}
						\end{subfigure}
						
						\medskip
						
						\begin{subfigure}[t]{.3\textwidth}
							\centering
							\includegraphics[width=\linewidth]{W2.png}
							\caption{$ W_2(x) $}
						\end{subfigure}
						%	\hfill
						\hspace{1cm}
						\begin{subfigure}[t]{.3\textwidth}
							\centering
							\includegraphics[width=\linewidth]{W3.png}
							\caption{$ W_3(x) $}
						\end{subfigure} 
					\end{center} 
					%}
				\caption{Walsh functions in the natural (or Hadamard) ordering for $ N=4 $. }\label{fig_walsh_four}
			\end{figure}
			
		\end{example}

		\begin{figure}[htp!]
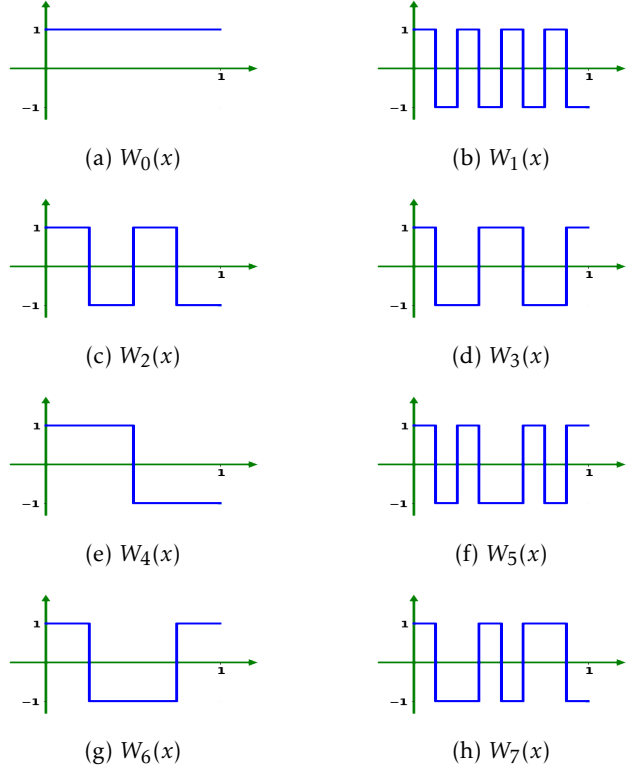

			\begin{center}
				\begin{subfigure}[t]{\mywidth}
					\centering
					\includegraphics[width=\linewidth,height=\myheight]{ww0.png}
					\caption{$ W_0(x) $}
				\end{subfigure}
				%	\hfill
				\hspace{1cm}
				\begin{subfigure}[t]{\mywidth}
					\centering
					\includegraphics[width=\linewidth,height=\myheight]{ww1.png}
					\caption{$ W_1(x) $}
				\end{subfigure}
				
				\medskip
				
				\begin{subfigure}[t]{\mywidth}
					\centering
					\includegraphics[width=\linewidth,height=\myheight]{ww2.png}
					\caption{$ W_2(x) $}
				\end{subfigure}
				%	\hfill
				\hspace{1cm}
				\begin{subfigure}[t]{\mywidth}
					\centering
					\includegraphics[width=\linewidth,height=\myheight]{ww3.png}
					\caption{$ W_3(x) $}
				\end{subfigure}

				\medskip
				
				\begin{subfigure}[t]{\mywidth}
					\centering
					\includegraphics[width=\linewidth,height=\myheight]{ww4.png}
					\caption{$ W_4(x) $}
				\end{subfigure}
				%	\hfill
				\hspace{1cm}
				\begin{subfigure}[t]{\mywidth}
					\centering
					\includegraphics[width=\linewidth,height=\myheight]{ww5.png}
					\caption{$ W_5(x) $}
				\end{subfigure}

				\medskip
				
				\begin{subfigure}[t]{\mywidth}
					\centering
					\includegraphics[width=\linewidth,height=\myheight]{ww6.png}
					\caption{$ W_6(x) $}
				\end{subfigure}
				%	\hfill
				\hspace{1cm}
				\begin{subfigure}[t]{\mywidth}
					\centering
					\includegraphics[width=\linewidth,height=\myheight]{ww7.png}
					\caption{$ W_7(x) $}
				\end{subfigure} 
				
			\end{center} 
			
			\caption{Walsh functions in the natural (or Hadamard) ordering for $ N=8 $.}
		\end{figure}

		We note that in $ \meqref{eq_def_walsh} $ the $ k^{\text{th}} $ Walsh function $ W_k(x) $ of order $ N $ is defined in terms of characters  $ \chi_{k} \in \widehat{G}  $, where $ G = (\ZZ/2\ZZ)^n = \{0,1\}^n $ and $ k =0,~1,~2, ~\ldots~,~ N-1 $. It is clear from the above discussion that the character theory of finite groups provides valuable insights into the structure of Walsh functions.

		\subsection{Walsh-Hadamard Transform}\label{subsec:wh_transform}
		
		Let $ G = (\ZZ/2\ZZ)^n = \{0,1\}^n   $ and $ f \in L(G) $.   Then from our earlier discussion (see Definition \ref{def_fourier_transform_def} and $ \meqref{eq_explicit_def_chik} $) 
		we have
		\begin{equation} \label{eq_def_walsh_transform}
			\widehat{f} (\chi_k) = \frac{1}{\sqrt{|G|}} \sum_{g \in G} \, f(g) \chi_k (g) 
			= \frac{1}{\sqrt{N}} \sum_{m = 0}^{2^n -1} \, f(m) \chi_k (m) 
			= \frac{1}{\sqrt{N}}  \sum_{m = 0}^{2^n -1} \, f(m) (-1)^{\sum_{i=0}^{n-1} \, m_i k_i }.
		\end{equation}
		where  $ k $ has the binary representation  $  k_{n-1} k_{n-2}k_1 k_{0} $  and $ m $ has the binary representation $  m_{n-1} m_{n-2}m_1 m_{0} $ and $ N= |G| = 2^n $. 
		Similarly, following \meqref{eq_ci}, we have
		\begin{equation}\label{eq_inverse_walsh_tansform}
			f(m) =  \frac{1}{\sqrt{N}}  \sum_{k = 0}^{N -1} \, \widehat{f}(\chi_k) \chi_k(m) = \frac{1}{\sqrt{N}}  \sum_{k = 0}^{N -1} \, \widehat{f}(\chi_k)  (-1)^{\sum_{i=0}^{n-1} \, m_i k_i }.
		\end{equation}

		\begin{defn}
			Let $ {\bf{v}} = [f(0) \quad  f(1) \quad  f(2) \quad  ~\ldots~ \quad f(N-1) ]^T$  be a vector with  $ N = 2^n $ components. Then its Walsh-Hadamard transform is the vector defined by
			$ {\bf{\widehat{v}}} = [\widehat{f} (\chi_0)\quad  \widehat{f} (\chi_1) \quad \widehat{f} (\chi_2) \quad ~\ldots~ \quad \widehat{f} (\chi_{N-1})]^T  $, where $ \widehat{f} (\chi_k)  $ is defined by \meqref{eq_def_walsh_transform}, for $ k=0,~1,~2,~\ldots~,~N-1 $. Similarly, given $ {\bf{\widehat{v}}} $, its inverse Walsh-Hadamard transform is $ {\bf{v}} $ such that the component $ f(m) $ of $ {\bf{v}} $ is defined by \meqref{eq_inverse_walsh_tansform}.
		\end{defn} 
		We will write \[
		[f(0) \quad f(1) \quad f(2) \quad ~\ldots~ \quad f(N-1) ]^T  \longleftrightarrow [\widehat{f} (\chi_0) \quad  \widehat{f} (\chi_1) \quad \widehat{f} (\chi_2) \quad ~\ldots~ \quad \widehat{f} (\chi_{N-1}) ]^T 
		\]
		to denote the pair of a vector and its Walsh-Hadamard transform.

		\begin{example}\label{ex_walsh_tranform}
			The Walsh-Hadamard transform of the vector
			$ {\bf{v}} = [f(0) \quad f(1) \quad f(2) \quad  f(3) ]^T $ is 
			
			$$ \widehat{\bf{v}} = \frac{1}{2}
			\begin{pmatrix*}[r]
				\, f(0) + f(1) + f(2) + f(3)  \, \\
				\, f(0) - f(1) + f(2) - f(3)  \, \\
				\, f(0) + f(1) - f(2) - f(3)  \, \\
				\, f(0) - f(1) - f(2) + f(3)  \,
			\end{pmatrix*}.
			$$
			One can compute this transform by computing 
			\[ \left(H \otimes H \right)  \, {\bf{v}} =   \, \frac{1}{2} \, \begin{pmatrix*}[r]
				$ 1 $ & $ 1 $  & $ 1 $  & $ 1 $  \\
				$ 1 $ & $ -1 $ & $ 1 $  & $ -1 $ \\ 
				$ 1 $ & $ 1 $ & $ -1 $ & $ -1 $ \\ 
				$ 1 $ & $ -1 $ & $ -1 $  & $ 1 $ \\ 
			\end{pmatrix*} \,  \begin{pmatrix*}[r]
				$ f(0) $  \\
				$ f(1) $  \\
				$ f(2) $  \\
				$ f(3) $ 
			\end{pmatrix*}. \] 
			
		\end{example}
		It can be checked that, Ex.~\ref{ex_walsh_tranform} generalizes for computing the Walsh-Hadamard transform of an input vector of of size $ N=2^n $. In fact, for any integer $ N=2^n $, the Walsh-Hadamard transform of $ {\bf{v}} = [f(0) \quad  f(1) \quad ~\ldots~ \quad  f(N-1) ]^T $ is given by 
		\begin{equation}\label{eq_hadamard_transform}
			(H^ {\otimes n}   \, \bf{v}).
		\end{equation}
		
		A naive approach to compute the Walsh-Hadamard transform involving  matrix-vector multiplication is of the order $ \mathcal{O}(N^2) $ where $ N=2^n $. 
		Although, in practice one computes the Walsh-Hadamard transform by employing a faster classical algorithm, namely the Fast Walsh-Hadamard Transform~\cite{geadah1977natural},~\cite{beauchamp1975walsh}. The classical Fast Walsh-Hadamard Transform  
		algorithm has the time complexity of the order of $ \mathcal{O} (N \log_2 (N)) $ for computing the Walsh-Hadamard transform of an input vector of size $ N=2^n $.

		\subsection{Quantum Walsh-Hadamard Transform}\label{subsec:quantum_wh_transform}
		
		The quantum implementation of Walsh-Hadamard transform involves two main steps. First, preparing the initial state  $ \sum_{k=0}^{N-1}\, f(k) \ket{k} $, and second applying Hadamard gates $ H^{\otimes n} $ on it. Here $ N=2^n $ and it is also assumed that $ \norm{f} =1 $, or equivalently, $ \sum_{k=0}^{N-1} \, (f(k))^2 = 1 $.
		It can be verified that,
		\[
		H^{\otimes n} \left[  \, \sum_{k=0}^{N-1}\, f(k) \ket{k} \right]   = \frac{1}{\sqrt{N}} \sum_{k=0}^{N-1}\,  \left(\sum_{m = 0}^{N -1} \, f(m) (-1)^{\sum_{i=0}^{n-1} \, m_i k_i }\right) \,\ket{k} =   \sum_{k = 0}^{N -1} \, \widehat{f} (\chi_k) \,\ket{k}. 
		\]

		\section{Hybrid Classical-Quantum Approach for Walsh-Hadamard Transform}\label{sec:hybrid-cq-wht}
		
		\begin{figure}[ht]
			\centering
			\includegraphics[width=0.5\linewidth]{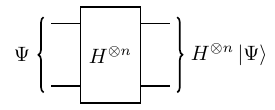}
			\caption{Quantum circuit for computing Walsh-Hadamard transform.}\label{fig_quantum_hadamard_n}
		\end{figure}

		A variation of quantum Walsh-Hadamard transform is at the heart of our algorithm to solve nonlinear differential and integral equations. 
		Indeed, the Hadamard gate is one of the most useful quantum gates and the Walsh-Hadamard transform is the first step in many important quantum algorithms. It was discussed earlier that the Walsh-Hadamard transform of $ {\bf{v}} = [f(0) \quad  f(1) \quad ~\ldots~ \quad  f(N-1) ]^T $ is given by 	$   (H^ {\otimes n}   \, \bf{v}) $. Assuming that  $ \norm{{\bf{v}}} =1 $, a simple quantum circuit consisting of $ n $ Hadamard gates can compute the Walsh-Hadamard  transform of an input vector $ \bf{v} $ of size $ N=2^n $ (see Fig.~$ \ref{fig_quantum_hadamard_n} $). However, the difficulty in this simple approach lies in the measurement. For example, consider the case of $ n=2 $ qubits with $  {\bf{v}} = [f(0) \quad  f(1) \quad f(2) \quad f(3) ]^T $. The circuit for computing the Walsh-Hadamard  transform $ \widehat{\bf{v}} $ in this case is shown in Figure~\ref{fig_quantum_hadamard}. The input state for the circuit in Figure~\ref{fig_quantum_hadamard} is prepared to be 
		\[
		{\bf{v}} = f(0) \ket{00} + f(1) \ket{01} + f(2) \ket{10} + f(3) \ket{11}.
		\]

		\begin{figure}[ht]
			\centering
			\includegraphics[width=0.5\linewidth]{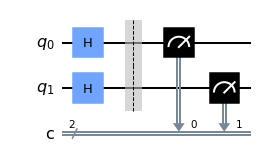}
			\caption{Quantum circuit for computing Walsh-Hadamard transform for $ n=2 $ qubits.}\label{fig_quantum_hadamard}
		\end{figure}
		However, one can only find the square of the amplitudes of the Walsh-Hadamard transform values by carrying out the measurement in this case. Since the input sequence is assumed to be real, the components of Walsh-Hadamard transformed vector $ \widehat{\bf{v}} $ would also be real. However, the  components of  $\, \widehat{\bf{v}} \, $  may be positive or negative and this sign information is lost on carrying out the measurement. We tackle this problem by using the following simple observation.

		\begin{lem}
			Let \[ [f(0) \quad f(1) \quad f(2) \quad ~\ldots~ \quad f(N-1) ]^T  \longleftrightarrow [\widehat{f} (\chi_0) \quad  \widehat{f} (\chi_1) \quad \widehat{f} (\chi_2) \quad ~\ldots~ \quad \widehat{f} (\chi_{N-1}) ]^T\] be a pair of a sequence and its Walsh-Hadamard transform. If
			\begin{equation}\label{eq_assumption}
				f(0) > \sum_{k=1}^{N-1} |f(k)|,
			\end{equation}
			then all the components of $ [\widehat{f} (\chi_0) \quad \widehat{f} (\chi_1) \quad \widehat{f} (\chi_2) \quad ~\ldots~ \quad \widehat{f} (\chi_{N-1}) ]^T $ are positive.
		\end{lem}
		\begin{proof}
			It is a consequence of the definition of Walsh-Hadamard transform that each component of the transformed vector $$ [\widehat{f} (\chi_0) \quad \widehat{f} (\chi_1) \quad \widehat{f} (\chi_2) \quad ~\ldots~ \quad \widehat{f} (\chi_{N-1}) ]^T $$ is of the form of a constant multiple of  $ f(0) \pm f(1) \pm f(2) \pm \quad ~\ldots~ \quad \pm f(N-1) $. The given hypothesis  $ f(0) > \sum_{k=1}^{N-1} |f(k)| $ then implies that all the components of the vector $$ [\widehat{f} (\chi_0) \quad \widehat{f} (\chi_1) \quad \widehat{f} (\chi_2) \quad ~\ldots~ \quad \widehat{f} (\chi_{N-1}) ]^T $$ are positive.
		\end{proof}
		
		\begin{example}
			The Walsh-Hadamard transform of $ {\bf{v}} = [f(0) \quad f(1) \quad f(2) \quad  f(3) ]^T $ is $$ \widehat{\bf{v}} = \frac{1}{2}
			\begin{pmatrix*}[r]
				\, f(0) + f(1) + f(2) + f(3)  \, \\
				\, f(0) - f(1) + f(2) - f(3)  \, \\
				\, f(0) + f(1) - f(2) - f(3)  \, \\
				\, f(0) - f(1) - f(2) + f(3)  \,
			\end{pmatrix*}.
			$$
			It is clear that if $  f(0) > |f(1)| + |f(2)| + |f(3)|  $, then all the components of the vector $ \widehat{\bf{v}} $ are positive.   In this case, one can directly compute the Walsh-Hadamard Transform by carrying out measurement in the circuit shown in Figure~\ref{fig_quantum_hadamard}. Since all the components of the transformed vector are positive, there is no ambiguity about signs when computing the components from the probability measurements.   
		\end{example}
	 
		From the above example, it follows  that 
		if the input vector $ {\bf{v}} =  [f(0) \quad f(1) \quad f(2) \quad ~\ldots~ \quad  f(N-1) ]^T  $ is normalized (i.e., $ \norm{\bf{v}} =1 $) and $ \meqref{eq_assumption} $ is satisfied, then it is easy to compute the Walsh-Hadamard transform. For example, the circuit shown in Figure~\ref{fig_quantum_hadamard} can be used for computing the Walsh-Hadamard transform for $ N=4 $.
		Of course, the condition in $ \meqref{eq_assumption} $ does not always hold. In such cases, we proceed as follows. Define  $  b_0 = \epsilon + \sum_{k=0}^{N-1} \, |f(k)| $, where $ \epsilon $ is an arbitrary positive number. Then the Walsh-Hadamard transform of $ {\bf{v_1}} = [b_0 \quad f(1) \quad  f(2) \quad ~\ldots~ \quad f(N-1) ]^T$ is given by 
		\begin{equation}\label{eq_shift_trick}
			\widehat{\bf{\, \, v_1}} = \widehat{\bf{v}} +  \widehat{\bf{t}}, \qquad \text{ where  $\widehat{\bf{t}} = \delta \, [ 1 \quad 1 \quad ~\ldots~ \quad 1  ]^T $ is vector of size $ N $ with $ \delta = \frac{1}{\sqrt{N}} (b_0 - f(0)) $.} 
		\end{equation}
		
		We note that vectors $ \bf{v} $ and ${\bf{v_1}} $ differ only in the first component. A quantum state is prepared based on the normalized state vector $ \frac{{\,\,\bf{v_1}}}{\norm{\bf{v_1}}} $.  Clearly, all components of the normalized state vector $\frac{\bf{v_1}}{\norm{\bf{v_1}}} $ are non-negative and the condition given in $ \meqref{eq_assumption} $ is satisfied. Therefore, the  Walsh-Hadamard transform $ \frac{\widehat{\,\,\bf{v_1}}}{\norm{\bf{v_1}}}  $ can be obtained using the quantum Hadamard gates and measurement operations without any ambiguity about the signs.  Then using $ \meqref{eq_shift_trick} $, one can compute $ \widehat{\bf{v}} = 	\widehat{\bf{\, \, v_1}} -  \widehat{\bf{t}} $. The hybrid classical-quantum approach explained above is employed in Algorithm~$ \ref{alg_QWHT} $ to compute the Walsh-Hadamard transform of a given input vector.

		\begin{algorithm}[H] \label{alg_QWHT}
			\DontPrintSemicolon
			\KwInput{The input vector $ A = [a_0 \quad a_1 \quad a_2 \quad  ~\ldots~ \quad a_{N-1} ]^{T} $ where $ N =2^n $ is a positive integer and $ a_i \in \RR $ for $ i=0 $ to $ i=N-1 $. }
			\KwOutput{The Walsh-Hadamard transform of the input vector.}
			%	\KwData{Testing set $x$}
			\Fn{$ \QWHT $ (A)}{
					$ b_0 = \epsilon + \sum_{k=0}^{N-1} \, |a_k| $ \tcp*{Here $ \epsilon $ is any positive number.}
					$   c = \sqrt{\left[ b_0^2 + \sum_{k=1}^{N-1} a_k^2 \right]}$ \tcp*{ Let $ \widetilde{A} = [b_0 \quad a_1 \quad a_2, ~\ldots~ \quad a_{N-1} ]^{T}$. Then $ c = \norm{\widetilde{A}} $ } 
					Prepare the state $ \ket{\Psi} = \frac{b_0}{c} \ket{0} + \sum_{k=1}^{N -1}\, \frac{a_k}{c} \ket{k}$ using $ n $ qubits. \tcp*{Initialize the state $ \ket{\Psi}  $ with $ \frac{\widetilde{A}}{\norm{\tilde{A}}}$.}
					Apply $ H^{\otimes } $ on $ \ket{\Psi} $. \\
					Measure all the $ n $ qubits to compute the probability $ p_k $ of obtaining the state $ \ket{k} $, for $ k=0 $ to $ 2^n-1 $. \\
					$ \delta = \frac{1}{\sqrt{N}}(b_0 - a_0 )$ \\
					\Return{ the vector $  [c\sqrt{p_0} - \delta \quad c\sqrt{p_1} - \delta \quad c\sqrt{p_2} - \delta \quad ~\ldots~ \quad c\sqrt{p_{N-1}} - \delta ]^{T} $}
					%}
			}
			\caption{A hybrid classical-quantum algorithm for computing the Walsh-Hadamard transform $ \QWHT(A) $ of a given input vector $ A $.}
		\end{algorithm}
		
		We note that the parameter $ \epsilon $ ensures that Algorithm~$ \ref{alg_QWHT} $ also works for the special case when the $ \norm{A} =0 $.

		\subsection{Computational Complexity}\label{sec:comp_complexity}

		The computational complexity of the classical Fast Walsh-Hadamard Transform~\cite{geadah1977natural} for an input vector of size $ N $ is of the order of $ \mathcal{O} (N \log_2(N)) $ additions and subtractions. Here we will show that our proposed hybrid classical-quantum algorithm for computing the Walsh-Hadamard transform has a complexity of $ \mathcal{O}(N) $.

		In order to make a more precise comparison of the performances of the classical and the hybrid algorithms, we consider an input vector of size $ N $ such that each of the $ N $ components is $ K $ bits, and provide a relevant estimate of the complexity. 
		Since the classical Fast Walsh-Hadamard Transform~\cite{geadah1977natural} for an input vector of size $ N $ is of the order of $ \mathcal{O} (N \log_2(N)) $ additions and subtractions and the addition or subtraction of two $ K $ bits numbers is of the order of $ \mathcal{O}(K) $ (using the elementary school algorithms of addition and subtraction with carry and borrow), respectively, the classical Fast Walsh-Hadamard Transform will be of the order of  $ \mathcal{O}(N \log_2(N) K)$. 
		
		Now we consider a breakdown of the computational complexity of  Algorithm~$ \ref{alg_QWHT} $ by considering the computational complexity of its individual steps.
		Note that step 2 in Algorithm~$ \ref{alg_QWHT} $ involves additions with associated complexity $  \mathcal{O}(NK) $. Step 3 involves computation of squares, additions and a square-root.   We note that the complexity of computing the square-root of a number is of the same order as that of multiplication using Newton's method. Since multiplication of two $ K $ bits number is $ \mathcal{O} (K\log_2(K)) $ as given by Harvey-Hoeven algorithm~\cite{harvey2021integer}, Step 3 has a complexity of  $ \mathcal{O}(N  K \log_2(K))$. Similarly, steps 7 and 8 put together have a complexity of $ \mathcal{O}(N  K \log_2(K))$. It is important to note that step 5 involves quantum Hadamard gates which enable the computation of  Hadamard transform with a complexity of $ \mathcal{O}(1) $. All the steps discussed so far (i.e., steps 2, 3, 5, 7 and 8) together have a complexity of $ \mathcal{O}(N  K \log_2(K))$. The cost of state preparation and measurement are often ignored in the complexity analysis of quantum algorithms. Here we will assume that the cost of state-preparation (step 4) and measurement (step 6) is bounded by $ \mathcal{O}(N  K \log_2(K))$. 
		Therefore, our proposed hybrid classical-quantum algorithm, Algorithm \ref{alg_QWHT}, is of the order of  $ \mathcal{O}(N  K \log_2(K))$. It follows that for a fixed  $ K $   our proposed algorithm is of the order of $ \mathcal{O}(N) $, whereas as mentioned earlier, the classical Fast Walsh-Hadamard Transform is of the order of $ \mathcal{O}(N \log_2(N)) $.

		\begin{remark}
			\begin{enumerate}[(a)] \,
				\item The Harvey-Hoeven algorithm for multiplication~\cite{harvey2021integer} is of order  $O (K\log_2(K)) $ only for extremely large numbers. The algorithm is  $O (K\log_2(K)) $ only if $ K \geq 2^{d^{12}} $ with $ d = 1729 $. Although they describe modifications in their algorithm, which will reduce $ d $ to $ 9 $. 
				\item Even if one uses the classical Karatsuba algorithm, which is of the order of $ \displaystyle \mathcal{O}(K^{1.585}) $, in computing  multiplications and square-roots in Algorithm \ref{alg_QWHT}, it is obvious that Algorithm \ref{alg_QWHT} will be of  $ \mathcal{O}(N) $.
			\end{enumerate}
		\end{remark}
		
		In the next section, using some simple examples, we illustrate how Walsh-Hadamard basis functions can be used for representing functions and performing integration and differentiation operations on them. This will set the stage for our proposed hybrid classical-quantum approach for the solution of nonlinear ordinary differential equations using Walsh-Hadamard basis functions and Algorithm \ref{alg_QWHT}.   
		
		\section{Walsh-Hadamard Basis: Representation, Integration and Differentiation of arbitrary functions}\label{sec:representation_wh_basis}

		Any piecewise continuous function $ f $ on $ [0,1] $ can be approximated by a discretizing the function and then the discretized version can be be expressed as a linear combination of Walsh functions. To demonstrate the use of Walsh-Hadamard transform for representation and integration of functions, especially in the context of the proposed hybrid classical-quantum approaches, we provide some illustrative examples. It is important to note that we consider the natural ordering of Walsh functions in our formulations unlike most works in the literature that are based on the sequency ordering. This consideration results in expressions for integration and differentiation matrices (as discussed below) that are slightly different from those commonly available in the literature (e.g., see \cite{beer1981walsh}, \cite{ahner1988walsh}, \cite{gnoffo2014global}, \cite{gnoffo2015unsteady}).
		\begin{example}
			Consider the function $ \cos (\pi t) $ on $ 0 \leq t \leq 1 $. The interval $ [0,1] $ is divided into $ N=4 $ sub-intervals and the function $ \cos (\pi t) $ is discretized as
			\[
			\cos (\pi t) \simeq \left[ \cos \left(  \frac{\pi}{8} \right) \quad \cos \left(\frac{3 \pi}{8}\right) \quad \cos \left( \frac{5\pi}{8} \right) \quad \cos \left(  \frac{7\pi}{8} \right) \right]^T .
			\]
			Then the computation of Walsh-Hadamard transform gives 
			\[
			\left[ \cos \left(  \frac{\pi}{8} \right) \quad  \cos \left(\frac{3 \pi}{8}\right) \quad \cos \left( \frac{5\pi}{8} \right) \quad \cos \left(  \frac{7\pi}{8} \right) \right]^T \longleftrightarrow \frac{1}{2}
			\left[0 \quad  1.08 \quad 2.61 \quad 0 \right]^T.
			\]
			It follows from the above calculation and and $ \meqref{eq_ci} $ (with $ |G|=N=4 $) that
			\[
			\cos (\pi t) \simeq \frac{1}{4} \left( 1.08 W_1(t) + 2.61 W_2(t) \right).
			\]

			Of course, increasing the value of $ N $ improves the approximations obtained using the Walsh functions. Figure~\ref{fig_walsh_approximation} below shows that graph of function 
			$ \cos(\pi t) $ for $ 0 \leq t \leq 1 $ and the approximations of this function obtained by using Walsh functions of order $ N=2$, $ 4 $, $ 8 $  and $ 16 $.
			
			\begin{figure}[ht]
				
				%	\resizebox{.9\linewidth}{!}{
					
					\begin{center}
						\begin{subfigure}[t]{.4\textwidth}
							\centering
							\includegraphics[width=\linewidth]{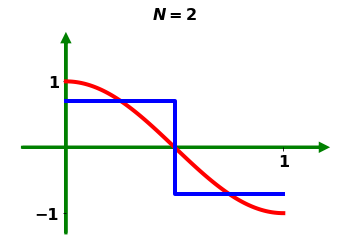}
						\end{subfigure}
						%	\hfill
						\hspace{1cm}
						\begin{subfigure}[t]{.4\textwidth}
							\centering
							\includegraphics[width=\linewidth]{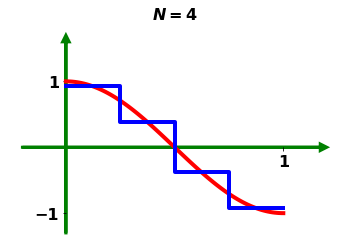}
						\end{subfigure}
						
						\medskip
						
						\begin{subfigure}[t]{.4\textwidth}
							\centering
							\includegraphics[width=\linewidth]{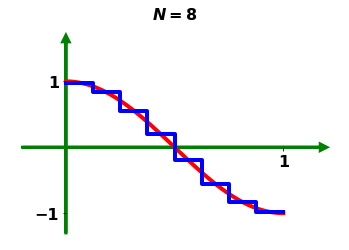}
						\end{subfigure}
						%	\hfill
						\hspace{1cm}
						\begin{subfigure}[t]{.4\textwidth}
							\centering
							\includegraphics[width=\linewidth]{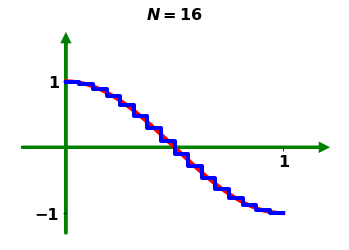}
						\end{subfigure} 
					\end{center} 
					%}
				\caption{Approximation of $ f(t) = \cos (\pi t) $, for $ 0 \leq t \leq 1 $, using Walsh functions of order $ N=2$, $ 4 $, $8$ and $ 16 $. The function $ \cos (\pi t) $ is plotted in red and the Walsh function based approximations are plotted in blue.}\label{fig_walsh_approximation}
			\end{figure}
			
		\end{example}

		\subsection{Integration of Walsh functions} \label{subSec_Integraiton_matrix}
		We have seen that Walsh functions can be used to approximate the discretized version of any continuous function defined on $ [0,1] $. One can use Walsh functions to approximate the integrals of piecewise continuous functions on $ [0,1] $. 
		For example, consider the function $ f(t) = \int_{0}^{t} W_0(x) \, dx  = t$, with $ 0 \leq t \leq 1 $. It can be discretized over $ [0,1] $, using $ N=4 $ as 
		\[ 
		f(t) = t \simeq \left[\frac{1}{8} \quad \frac{3}{8} \quad \frac{5}{8} \quad \frac{7}{8} \right]^T,
		\]
		and the Walsh-Hadamard transform computation gives
		\[
		t \simeq \left[\frac{1}{8} \quad \frac{3}{8} \quad \frac{5}{8} \quad \frac{7}{8} \right]^T \longleftrightarrow   \frac{1}{2} \left[2 \quad -\frac{1}{2} \quad -1 \quad 0 \right]^T.
		\] 
		From the above computation and $ \meqref{eq_ci} $ (with $ |G| =N =4$) it follows that 
		\begin{equation}\label{eq_walsh_integration}
			\int_{0}^{t} W_0(x) \, dx  = t = \frac{1}{4} \left( 2 W_0(t) - \frac{1}{2} W_1(t) - W_2(t)  \right) =\frac{1}{2} W_0(t) - \frac{1}{8} W_1(t)  - \frac{1}{4} W_2(t) .
		\end{equation}
		Similarly, it can be checked that 
		\begin{align*}
			\int_{0}^{t} W_1(x) \, dx  \simeq \left[ \frac{1}{8} \quad  \frac{1}{8} \quad \frac{1}{8} \quad \frac{1}{8} \right]^T \longleftrightarrow  \frac{1}{2} \left[ \frac{1}{2} \quad 0 \quad 0 \quad 0\right]^T.
		\end{align*}
		This gives,
		\begin{equation}\label{eq_walsh_integration_one}
			\int_{0}^{t} W_1(x) \, dx  =  \frac{1}{8} W_0(t).
		\end{equation}
		We have 
		\[
		\int_{0}^{t} W_2(x) \, dx  \simeq \left[ \frac{1}{8} \quad \frac{3}{8} \quad \frac{3}{8} \quad \frac{1}{8}\right]^T \longleftrightarrow \frac{1}{2} \left[ 1 \quad 0 \quad 0 \quad -\frac{1}{2}\right]^T,
		\]
		and 
		\[
		\int_{0}^{t} W_3(x) \, dx  \simeq \left[ \frac{1}{8} \quad \frac{1}{8} \quad -\frac{1}{8} \quad -\frac{1}{8}\right]^T \longleftrightarrow \frac{1}{2} \left[0 \quad 0 \quad \frac{1}{2} \quad 0 \right],
		\]
		and similar calculations give,
		\begin{align}
			\int_{0}^{t} W_2(x) \, dx  &= \frac{1}{4} W_0(t)  - \frac{1}{8} W_3(t), \\
			\int_{0}^{t} W_3(x) \, dx &= \frac{1}{8} W_2(t).
		\end{align}
		
		The above calculations give rise to the following \textit{integration matrix}
		\renewcommand{\arraystretch}{1.25}
		\begin{equation}\label{eq_integration_matrix}
			I_4  = 
			\begin{pmatrix}
				\frac{1}{2}	& \frac{1}{8} & \frac{1}{4} &  0\\ 
				-\frac{1}{8}& 0 & 0 & 0  \\
				-\frac{1}{4}& 0 & 0 & \frac{1}{8} \\
				0 &  0 & -\frac{1}{8} & 0
			\end{pmatrix}.
		\end{equation}
		The function to be integrated can be represented as a column vector with the Walsh functions as the basis functions. Then the multiplying this column vector by the integration matrix is equivalent to performing the integration of the function. 
		\begin{example}
			The function $\cos t$ can be discretized and represented as $ \left[ \cos\left(\frac{1}{8}\right) \quad \cos\left(\frac{3}{8}\right) \quad \cos\left(\frac{5}{8}\right) \quad \cos\left(\frac{7}{8}\right)\right] $ on $ [0,1] $ for $ N=4 $. 
			Performing the Walsh-Hadamard transform gives
			\[
			\left[ \cos\left(\frac{1}{8}\right)\quad \cos\left(\frac{3}{8}\right) \quad \cos\left(\frac{5}{8}\right) \quad \cos\left(\frac{7}{8}\right)\right] \longleftrightarrow \frac{1}{2} \left[3.375 \quad 0.232 \quad 0.471 \quad -0.108\right].
			\]
			This means 
			\[
			\cos t  \simeq \frac{1}{4} \left(  3.375 W_0(t) +  0.232 W_1(t) +  0.471 W_2(t) -0.108 W_3(t) \right).  
			\]
			Therefore the column coordinate vector for $ \cos t $ in terms of Walsh basis vectors is given by
			\[
			{\bf{v}} = \left[0.844 \quad 0.058 \quad 0.118 \quad -0.027\right]^T,
			\]
			and 
			\[
			I_4 {\bf{v}} = \begin{pmatrix}
				\frac{1}{2}	& \frac{1}{8} & \frac{1}{4} &  0\\ 
				-\frac{1}{8}& 0 & 0 & 0  \\
				-\frac{1}{4}& 0 & 0 & \frac{1}{8} \\
				0 &  0 & -\frac{1}{8} & 0
			\end{pmatrix} \begin{pmatrix}
				0.844  \\
				0.058  \\
				0.118 \\
				-0.027
			\end{pmatrix} = \begin{pmatrix}
				0.459  \\
				-0.105  \\
				-0.214  \\
				-0.015 
			\end{pmatrix}.  
			\] 
			Therefore, we have
			\[
			\int_{0}^{t} \, \cos x \, dx \simeq 0.459 W_0(t) -0.105 W_1(t) -0.214  W_2(t) -0.015 W_3(t).
			\]
			One can transform this back to obtain the discretization of $ \int_{0}^{t} \, \cos x \, dx  = \sin t $ in the time domain as follows.  \[
			2 (H \otimes H) \begin{pmatrix}
				0.459  \\
				-0.105  \\
				-0.214  \\
				-0.015 
			\end{pmatrix} =  \begin{pmatrix}
				0.125  \\
				0.366\\
				0.585  \\
				0.767 
			\end{pmatrix} = \begin{pmatrix}
				\sin (\frac{1}{8})  \\
				\sin (\frac{3}{8}) \\
				\sin (\frac{5}{8}) \\
				\sin (\frac{7}{8}) 
			\end{pmatrix}.
			\]
		\end{example}
		A similar calculation shows that for  $ N=8 $ the integration matrix $ I_{8} $ is given by
		\[
		I_8 = \begin{pmatrix}  \frac{1}{2} & \frac{1}{16} & \frac{1}{8} & 0 & \frac{1}{4} & 0 & 0 & 0\\  -\frac{1}{16} & 0 & 0 & 0 & 0 & 0 & 0 & 0\\  -\frac{1}{8} & 0 & 0 & \frac{1}{16} & 0 & 0 & 0 & 0\\  0 & 0 & -\frac{1}{16} & 0 & 0 & 0 & 0 & 0\\ -\frac{1}{4} & 0 & 0 & 0 & 0 & \frac{1}{16} & \frac{1}{8} & 0\\  0 & 0 & 0 & 0 & -\frac{1}{16} & 0 & 0 & 0\\  0 & 0 & 0 & 0 & -\frac{1}{8} & 0 & 0 & \frac{1}{16}\\  0 & 0 & 0 & 0 & 0 & 0 & -\frac{1}{16} & 0\\ \end{pmatrix}.
		\]
		The integration matrix $ I_{N} $ for $ N=2^n $ can be be computed similarly. The integration matrix $ I_{N} $ is a very sparse matrix and therefore computation of integration is not very costly even when $ n $ is large. It is easy to see that if $ f(t) $ is discretized as $ {\bf{f}} = [f\left(\frac{1}{2N}\right) \quad f\left(\frac{3}{2N}\right) \quad f\left(\frac{5}{2N}\right) \quad ~\ldots~ \quad f\left(\frac{2N-1}{2N}\right)  ]^T $, then  $ \int_{0}^{t} \, f(x) \, dx $ in the discretized form can be computed as
		\begin{equation}\label{eq_integration_new}
			% \sqrt{N}  \QWHT( I_N \frac{1}{\sqrt{N}} \QWHT({\bf{f}} )) = 
			H^{\otimes n} \, ( I_N \, H^{\otimes n} \, ({\bf{f}} )),
		\end{equation}
		for $  0 \leq t \leq 1. $
		In our proposed approach based on a hybrid classical-quantum implementation of the Walsh-Hadamard transform as outlined in Algorithm~$ \ref{alg_QWHT} $, the above expression can be evaluated as 
		\begin{equation}\label{eq_integration}
			\QWHT( I_N \QWHT({\bf{f}} )).
		\end{equation}
		
		Although we have provided examples in the domain $ [0, \, 1] $, 
		the present approach can easily be rescaled for arbitrary domains $ [t_{\text{initial}}, \,t_{\text{final}} ] $.

		\subsection{Differentiation Matrix}\label{subsec:differentiation_matrix}
		Since differentiation can be viewed as an inverse process of integration, the \textit{differentiation matrix} can be obtain by finding the inverse of the integration matrix.
		For example, for $ N=4 $, the differentiation matrix $ D_4 $ is given by
		\[
		D_4 = I_4^{-1} = \begin{pmatrix}
			0	& -8 & 0 &  0\\ 
			8& 32 & 0 & 16 \\
			0& 0 & 0 & -8 \\
			0 &  -16 & 8 & 0
		\end{pmatrix}.
		\]
		Similarly, the differentiation matrix $ D_{N} $ of order $ N=2^n $ can be be computed as the inverse of $ I_{N} $. The differentiation matrix $ D_{N} $ is also a sparse matrix and therefore computation of differentiation is also not very costly even when $ n $ is large. 
		
		\section{A Hybrid Algorithm for Solving Differential Equations}\label{sec:hybrid_algorithm_de}
		Walsh functions and the Walsh-Hadamard transform can be used to solve nonlinear differential equations. These methods are well-known and were used by several authors for solving differential equations~\cite{gnoffo2014global},~\cite{gnoffo2015unsteady},~\cite{beer1981walsh}. The key idea of our work is to replace the classical computation of Walsh-Hadamard transform with a hybrid classical-quantum computation of Walsh-Hadamard transform (based on the natural ordering). Since the hybrid classical-quantum version of Walsh-Hadamard transform can be carried out more efficiently (in $ (\mathcal{O}(N)) $ operations) than the classical Fast Walsh-Hadamard transform ($ \mathcal{O}(N \log_2(N)) $), it provides a superior algorithm to solve nonlinear differential equations.
		
		Consider the following system of nonlinear ordinary differential equations  
		\begin{align} 
			\frac{d x_i}{dt} &= f_i(x_1, x_2, ~\ldots~ ,x_m, t), \quad \text{for $ i =1,2, ~\ldots~, m$,}\label{eq_ode_general_one} 
		\end{align}
		with the initial conditions
		\begin{equation}\label{eq_ode_initial}
			x_i(0) =  q_i.
		\end{equation}

		\begin{algorithm}[H] \label{alg_main_general}
			\DontPrintSemicolon
			\KwInput{Integers $ N  $  and $ N_{max} $.}
			\KwOutput{Discretized versions of $ x_i(t) $, with $ N $ sub-intervals on $ 0 \leq t \leq 1 $, for $ i=1,2,~\ldots~,m $.}
			%	\KwData{Testing set $x$}
			\tcc{The algorithm uses the quantum subroutine $ \QWHT $($ \bf{v} $) to compute the quantum Walsh-Hadamard transform of the input vector $ \bf{v} $. }
			$ t =  [\frac{1}{2N}\quad \frac{3}{2N} \quad \frac{5}{2N} \quad ~\ldots~ \quad \frac{2N-1}{N}  ]^{T} $ \tcp*{Initialize $ t $.} 
			\For{$ i \gets 1$ \KwTo $ m$ }{
				$x_i = {x_{i}}_{\text{Initial}} = [q_i \quad q_i \quad q_i ~\ldots~ \quad q_i]^T $ \tcp*{Initialize $ x_i $.}}
			$ I_N =  \text{ The integration matrix of order } N $ \tcp*{Initialize the integration matrix $ I_N $.} 
			\For{$i\gets 1$ \KwTo $N_{max}$ }{
				$ x_i = {x_{i}}_{\text{Initial}} + \QWHT \left(I_N \, \QWHT(f_i(x_1, x_2, ~\ldots~ ,x_m, t ))  \right) $  \tcp*{Compute $ x_i = {x_{i}}_{\text{Initial}} + \int_{0}^{t}\, f_i  \, d\tau $.}        
			}
			\Return{$ x_1 ,  x_2 ,  ~\ldots~, x_m $.}
			\caption{A hybrid classical-quantum algorithm for solving the system of differential equation given in $ \meqref{eq_ode_general_one} $.}
		\end{algorithm}
		
		Our algorithm starts by initializing  $ {x_{i}} $ to a vector ${x_{i}}_{\text{Initial}} = [q_i \quad q_i \quad q_i \quad ~\ldots~ \quad q_i]^T $ of size $ N $ for $ i=1 $ to $ i=m $. Then using 
		$ \meqref{eq_ode_general_one} $ and $ \meqref{eq_ode_initial} $  one can write
		\begin{align*}
			x_i &= {x_{i}}_{\text{Initial}} + \int_{0}^{t}\, f_i(x_1, x_2, ~\ldots~ ,x_m,\tau ) \, d \tau 
		\end{align*}
		The Walsh functions based representation and the integration matrix approach, discussed earlier in Section~\ref{sec:representation_wh_basis}, can be used to solve the above equations iteratively. 
		The iterations can then be continued up to a maximum chosen number of times, $ N_{max} $, where $ N_{max} $ depends on the accuracy of the solution desired. One can select $ N_{max} $ such that the successive iterations give identical results to the required number of decimal places. 
		We note that the statement,  $$ x_i = {x_{i}}_{\text{Initial}} + \QWHT \left(I_N \, \QWHT(f_i(x_1, x_2, ~\ldots~ ,x_m,t )) \right) $$ in Algorithm~$ \ref{alg_main_general} $, computes  $ x_i = {x_{i}}_{\text{Initial}} + \int_{0}^{t}\, f_i(x_1, x_2, ~\ldots~ ,x_m,t ) \, d \tau $ (see \meqref{eq_integration} and the discussion before that). 
		Here, we also note that $ f_i(x_1, x_2, ~\ldots~ ,x_m, t ) $ is computed point-wise in the time-domain.

		\section{Computational Examples}\label{sec:computational_examples}
		
		In this section, we will give computational examples to illustrate Algorithm \ref{alg_main_general} for solving nonlinear differential equations. The proposed algorithms for the solution of nonlinear ordinary differential equations were successfully implemented and tested using the simulated environment of Qiskit (IBM's open source quantum computing platform).  We note the three examples given below illustrate our solution approach for Initial Value Problems (IVPs). The first and the second examples illustrate the proposed approach for solutions of non-stiff and stiff nonlinear ordinary differential equations in 1D. The third example demonstrates the extension to solution of multidimensional system of  nonlinear ordinary differential equations. 
		A brief discussion on extensions to boundary value problems is provided at the end of this section.
		
		\subsection{Example 1 - Riccati differential equation}
		We consider the following Riccati differential equation for $ 0 \leq t \leq 1 $,
		\begin{align} 
			\frac{d x_1}{dt} &= x_1^2 + x_1 + 1,
			\label{eq_ode_Riccati} 
		\end{align}
		with the initial conditions $ x_1(0) = - \frac{1}{2}  $. The analytical solution for the above initial value problem is 
		\[
		x_1 = \frac{1}{2} \left( \sqrt{3} \tan\left(\frac{\sqrt{3}t}{2}  \right) -1 \right).
		\]
		Our algorithm starts by initializing $x_1 ={x_1}_{\text{Initial}} = -\frac{1}{2} [ 1 \quad 1 \quad 1 \quad 1 ]^T $. Then using 
		$ \meqref{eq_ode_Riccati} $  one can write
		\begin{align}\label{eq_example_one}
			x_1(t) &=  {x_1}_{\text{Initial}} + \int_{0}^{t}\, f_1(x_1(\tau)) \, d\tau, 
		\end{align}
		where $ f_1(x_1(\tau)) = x_1(\tau)^2 + x_1(\tau) + 1 $.
		The last expression can now be used iteratively, employing the Walsh functions based representation and the integration matrix approach as discussed earlier in Section \ref{subSec_Integraiton_matrix}.
		The iterations can then be continued until the desired accuracy is achieved. The maximum number of iterations $ N_{max} $ can be chosen so that the successive iterations produce identical results up to the desired number of decimal places.  The value $ N_{max} $ is chosen to be $ 10 $ in Algorithm~$ \ref{alg_Riccati} $.
		The statement 	$ x_1 = {x_1}_{\text{Initial}} + \QWHT \left(I_4 \, \QWHT( x_1^2 + x_1 + 1) \right) $  in Algorithm~$ \ref{alg_Riccati}$ computes 
		$ x_1(t)$ by evaluating the right side of $ \meqref{eq_example_one} $.
		After the first iteration the result obtained is 
		\[
		x_1=  [-0.40625 \quad -0.21875 \quad -0.03125 \quad 0.15625]^T.
		\]
		The result obtained after the $ 10 $-th iteration, and rounded to the five decimal places, is 
		\[
		x_1 = [-0.40512 \quad -0.20567 \quad 0.02743 \quad 0.33735]^T.
		\]
		The correct analytic solution $ x_1(t) $ along with solutions obtained by Walsh functions of order $ N=2$, $ 4 $, $8 $, and $ 16 $ are shown in Figure~$ \ref{fig_walsh_approximation_solution_example_Riccati} $.
		
		\begin{figure}[ht]
			
			%	\resizebox{.9\linewidth}{!}{
				
				\begin{center}
					\begin{subfigure}[t]{.4\textwidth}
						\centering
						\includegraphics[width=\linewidth]{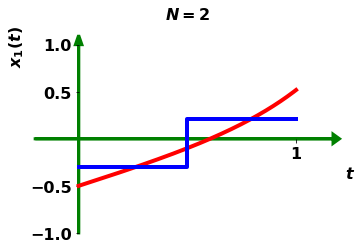}
					\end{subfigure}
					\hfill
					\hspace{1cm}
					\begin{subfigure}[t]{.4\textwidth}
						\centering
						\includegraphics[width=\linewidth]{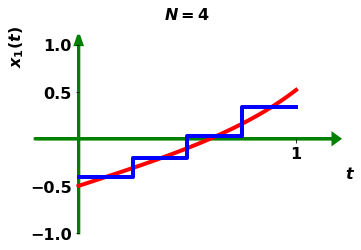}
					\end{subfigure}
					\medskip	
					\begin{subfigure}[t]{.4\textwidth}
						\centering
						\includegraphics[width=\linewidth]{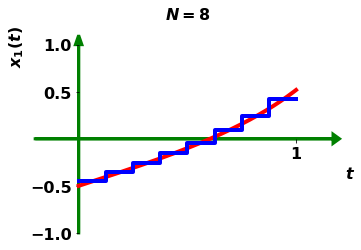}
					\end{subfigure}
					\hfill
					\hspace{1cm}
					\begin{subfigure}[t]{.4\textwidth}
						\centering
						\includegraphics[width=\linewidth]{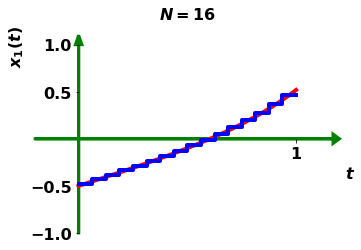}
					\end{subfigure} 
				\end{center} 
				%}
			\caption{Analytic solution for $ x_1(t) $ (in red) and the solutions obtained by Walsh functions of order $ N=2$, $ 4 $, $8 $ and $ N=16$ (in blue), based on the proposed hybrid classical-quantum approach. }\label{fig_walsh_approximation_solution_example_Riccati}
		\end{figure}

		\begin{algorithm}[H] \label{alg_Riccati}
			\DontPrintSemicolon
			
			\KwInput{$ N = 4 $  and $ N_{max} =10 $.}
			\KwOutput{Discretized versions of $ x_1(t) $ on $ 0 \leq t \leq 1 $ with $ N =4 $ sub-intervals.}
			%	\KwData{Testing set $x$}
			\tcc{The algorithm uses the quantum subroutine $ \QWHT $($ \bf{v} $) to compute the quantum Walsh-Hadamard transform of the input vector $ \bf{v} $. }
			$ t =  [\frac{1}{8}\quad \frac{3}{8} \quad \frac{5}{8} \quad  \frac{7}{8}  ]^T $ \tcp*{Initialize $ t $.} 
			$ x_1 = {x_1}_{\text{Initial}} = [-\frac{1}{2} \quad -\frac{1}{2} \quad -\frac{1}{2} \quad -\frac{1}{2}]^T $ \tcp*{Initialize $ x_1 $.}
			$ I_4 = \begin{pmatrix}
				\frac{1}{2}	& \frac{1}{8} & \frac{1}{4} &  0\\ 
				-\frac{1}{8}& 0 & 0 & 0  \\
				-\frac{1}{4}& 0 & 0 & \frac{1}{8} \\
				0 &  0 & -\frac{1}{8} & 0
			\end{pmatrix} $  \tcp*{Initialize the integration matrix $ I_4 $.} 
			\For{$i\gets 1$ \KwTo $N_{max}$ }{
				\tcc{ Compute $ x_1 = {x_1}_{\text{Initial}} + \int_{0}^{t}\, x_1(\tau)^2 + x_1(\tau) + 1 \, d\tau $.}
				$ x_1 = {x_1}_{\text{Initial}} + \QWHT \left(I_4 \, \QWHT( f_1 ) \right) $, where $ f_1(x_1) = x_1^2 + x_1 + 1 $          
			}        
			\Return{$ x_1 $.}
			
			\caption{A hybrid algorithm for solving the differential equation given in $ \meqref{eq_ode_Riccati} $.}
		\end{algorithm}

		\subsection{Example 2 - A stiff nonlinear ODE}\label{subsection_flame}
		
		To demonstrate the applicability of our proposed hybrid classical-quantum approach to solution of stiff, nonlinear ordinary differential equations, we consider a model of flame propagation as our next example.
		If one lights a match, initially the flame grows very quickly before settling down to a stable size. 
		The stable size is reached when the amount of oxygen available through the surface of the flame equals the amount of oxygen being consumed in the combustion process inside the flame ball. 
		This evolution is modeled by the following (stiff) nonlinear ordinary differential equation
		\begin{align} 
			\frac{dx_1}{d\tau} &= x_1^2 (1 - x_1), \label{flame_eq_ode} 
		\end{align}
		with the initial conditions $ x_1(0) = \delta  $ and $ \tau \in [0, 2/\delta]$. Here $ x_1(\tau) $ represents the radius of the flame ball at time $ \tau $. 
		
		It is convenient to introduce a change of the variable $ \tau $ as $ t = \frac{\delta}{2} \tau $ to rescale the time-domain to $ [0,1] $. Then the resulting initial value problem is   
		\begin{align} 
			\frac{d x_1}{dt} &= \frac{2}{\delta} \, x_1^2 (1 - x_1), \label{flame_eq_ode_main} 
		\end{align}
		with the initial conditions $ x_1(0) = \delta  $ and $ t \in [0, 1]$. 
		One can now write
		\begin{align}
			x_1(t) &=  {x_1}_{\text{Initial}} + \int_{0}^{t}\, f_1(x_1(s)) \, ds, \label{flame_eq_ode_two}
		\end{align}
		where $ f_1(x_1(\tau)) = \frac{2}{\delta} \, x_1(\tau)^2 (1 - x_1(\tau)) $ with
		$x_1 ={x_1}_{\text{Initial}} = \delta [ 1 \quad 1 \quad ~\ldots~ \quad ~\ldots~ \quad 1]^{T} $.  
		The approach to solve this problem is quite similar to the previous example and hence we skip the details here. The solution obtained using the Walsh functions of order $ N=2^8 $ and with $ \delta = 0.002 $ is plotted in Figure~$ \ref{Fig_flame} $ along with the exact solution to \meqref{flame_eq_ode_main}. We note that the exact solution to \meqref{flame_eq_ode_main} is 
		\begin{equation}\label{eq_exact_solution}
			x_1(t) = \frac{1}{W \left(a e^{a - \frac{2t}{\delta}}\right)  + 1},
		\end{equation}
		where $ a = 1 /\delta -1 $ and $ W $ is the Lambert $ W $ function (See Sec. \ 7.1, \cite{scheffel2012spectral}). 
		
		As noted in Figure~$ \ref{Fig_flame} $, the numerical solution to \meqref{flame_eq_ode_main}, based on the proposed hybrid classical-quantum approach using the Walsh functions, agrees quite well with the exact solution. 
		It is evident that the proposed numerical approach is able to accurately capture the sharp transition (that occurs at $t = 0.5$ in the radius of the flame.
		As noted in \cite{scheffel2012spectral}, explicit finite difference methods will require extremely small time steps to capture this sharp gradient.

		\begin{figure}[hpt]
			\centering
			\includegraphics[width=0.6\linewidth]{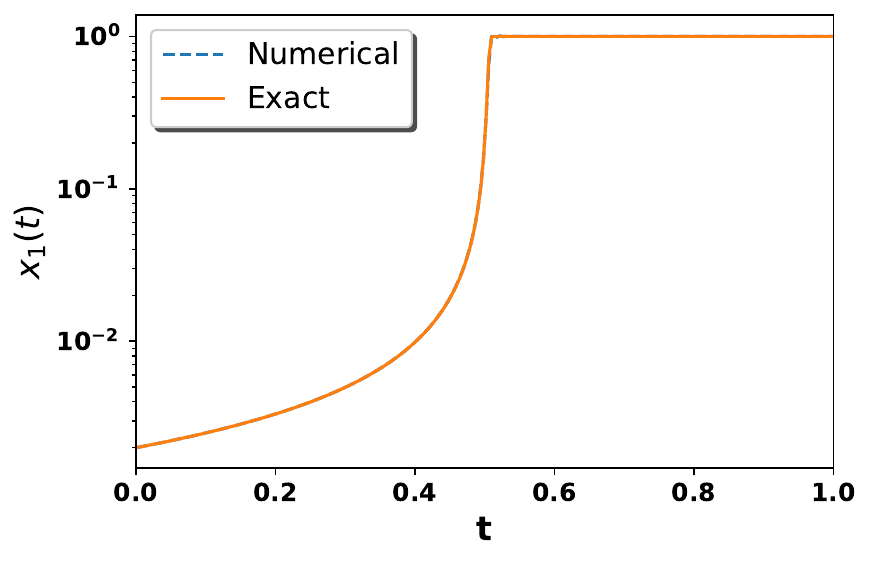}
			\caption{Numerical solution of a stiff nonlinear ODE (representing a model for flame propagation),   \meqref{flame_eq_ode_main}, based on the proposed hybrid classical-quantum approach using the Walsh functions of order $ N=2^8 $ is shown. Numerical solution of \meqref{flame_eq_ode_main} is in very good agreement with the exact solution of this equation (as given in \meqref{eq_exact_solution}).}
			\label{Fig_flame}
		\end{figure}

		\subsection{Example 3 - A system of nonlinear ODEs}
		In this section, we will demonstrate the application of our proposed hybrid classical-quantum approach to the solution of a system of nonlinear ordinary differential equations involving two dependent variables.  
		Consider the following system of differential equations~\cite{beer1981walsh}
		\begin{align} 
			\frac{d x_1}{dt} &= x_2, \label{eq_ode_one} \\ 
			\frac{d x_2}{dt} &= -(3x_1x_2 + x_1^3), \label{eq_ode_two}
		\end{align}
		with the initial conditions $ x_1(0) =0  $ and $ x_2(0) =1 $.

		\begin{algorithm}[H] \label{alg_main}
			\DontPrintSemicolon
			
			\KwInput{$ N = 4 $  and $ N_{max} =20 $.}
			\KwOutput{Discretized versions of $ x_1(t) $ and $ x_2(t) $ on $ 0 \leq t \leq 1 $ with $ N =4 $ sub-intervals.}
					\tcc{The algorithm uses the quantum subroutine $ \QWHT $($ \bf{v} $) to compute the quantum Walsh-Hadamard transform of the input vector $ \bf{v} $. }
			$ t =  [\frac{1}{8}\quad \frac{3}{8} \quad \frac{5}{8} \quad  \frac{7}{8}  ]^T $ \tcp*{Initialize $ t $.} 
			$x_1 = {x_{1}}_{\text{Initial}} = [0 \quad 0 \quad 0 \quad 0]^T $ \tcp*{Initialize $ x_1 $.}
			$ x_2 = {x_{2}}_{\text{Initial}}  = \left[1 \quad 1 \quad 1 \quad 1\right]^T  $ \tcp*{Initialize $ x_2 $.}
			$ I_4 = \begin{pmatrix}
				\frac{1}{2}	& \frac{1}{8} & \frac{1}{4} &  0\\ 
				-\frac{1}{8}& 0 & 0 & 0  \\
				-\frac{1}{4}& 0 & 0 & \frac{1}{8} \\
				0 &  0 & -\frac{1}{8} & 0
			\end{pmatrix} $  \tcp*{Initialize the integration matrix $ I_4 $.} 
			\For{$i\gets 1$ \KwTo $N_{max}$ }{
				\tcc{ Compute $ x_1 = {x_{1}}_{\text{Initial}} + \int_{0}^{t}\, x_2(\tau) \, d\tau $.}
				$ x_1 = {x_{1}}_{\text{Initial}} + \QWHT \left(I_4 \, \QWHT(f_1)  \right) $, where $ f_1(x_1,x_2) = x_2 $     \\        
				\tcc{ Compute $ x_2 = {x_{2}}_{\text{Initial}} - \int_{0}^{t}\,  \left(3x_1(\tau) x_2(\tau) + x_1(\tau)^3\right) \, d\tau $.}
				$ x_2 = {x_{2}}_{\text{Initial}}  + \QWHT \left(I_4 \, \QWHT( f_2 ) \right)$, where $ f_2(x_1,x_2) =  -(3x_1x_2 + x_1^3) $ 
			}
			\Return{$ x_1 $, $ x_2 $}
			
			\caption{A hybrid algorithm for solving the system of differential equation given in $ \meqref{eq_ode_one} $ and $ \meqref{eq_ode_two} $.}
		\end{algorithm}

		Our algorithm begins by initializing $x_1 = {x_{1}}_{\text{Initial}} = [0 \quad 0 \quad 0 \quad 0]^T $ and $x_2 = W_0(t) = \left[1 \quad 1 \quad 1 \quad 1\right]^T $. Then using 
		$ \meqref{eq_ode_one} $ and $ \meqref{eq_ode_two} $  one can write
		\begin{align*}
			x_1 &= {x_{1}}_{\text{Initial}} + \int_{0}^{t}\, f_1(x_1(\tau),x_2(\tau)) \, d \tau, \\
			x_2 &= {x_{2}}_{\text{Initial}} + \int_{0}^{t}\, f_2(x_1(\tau),x_2(\tau))  \, d \tau,
		\end{align*}
		where $ f_1(x_1(\tau),x_2(\tau)) = x_2(\tau) $ and $ f_2(x_1(\tau),x_2(\tau)) = - (3x_1(\tau) x_2(\tau) + x_1(\tau)^3) $.
		The last two equations can be solved using the Walsh functions based representation and the integration matrix approach, as discussed earlier in Section \ref{subSec_Integraiton_matrix}.
		The iterations can then be continued up to a maximum chosen number of times, $ N_{max} $, which depends on the accuracy of the solution desired. For example, $ N_{max} $ can be chosen such that the successive iterations give identical results to the required number of decimal places. 
		At the end of the first iteration we obtain
		\begin{align*}
			x_1 &\simeq \left[0.125 \quad 0.375 \quad 0.625 \quad 0.875\right]^T, \\
			x_2 &\simeq \left[1 \quad 1 \quad 1 \quad 1\right]^T.
		\end{align*}
		We note that in Algorithm~$ \ref{alg_main} $ the integration  $ x_1 = {x_{1}}_{\text{Initial}} + \int_{0}^{t}\, f_1(x_1(\tau),x_2(\tau)) \, d \tau $ is performed by the statement  $$x_1 = {x_{1}}_{\text{Initial}} + \QWHT \left(I_4 \, \QWHT(f_1)  \right). $$   
		Similarly, the statement 	$$ x_2 = {x_{2}}_{\text{Initial}}  + \QWHT \left(I_4 \, \QWHT( f_2 ) \right)$$ in Algorithm~$ \ref{alg_main} $ is used to compute $ x_2 = {x_{2}}_{\text{Initial}} - \int_{0}^{t}\,  f_2(x_1(\tau),x_2(\tau)) \, d\tau $. Here, we also note that computation $ f_2(x_1,x_2) = -(3x_1x_2 + x_1^3) $ is carried out point-wise in the time-domain. 
		For example, when $$ x_1 = \left[0.125 \quad  0.375 \quad  0.625 \quad  0.875\right]^T  \text{ and }  x_2 = \left[1 \quad 1 \quad 1 \quad 1\right]^T, $$ we have $$ x_1^3 = \left[0.001953125 \quad  0.052734375 \quad  0.244140625 \quad  0.669921875\right]^T,$$  $$3x_1 x_2 = \left[ 0.375 \quad  1.125 \quad  1.875 \quad  2.625  \right]^T, $$ and
		$$ -(3x_1x_2 + x_1^3) = \left[-0.376953125\quad  -1.177734375 \quad  -2.119140625 \quad  -3.294921875\right]^T.$$ 
		After $ 8 $-th iteration the result is 
		\begin{align*}
			x_1&=  [0.11960814 \quad  0.33997528 \quad  0.51224524 \quad  0.62590886]^T, \\
			x_2&=  [0.95686836 \quad  0.80607053 \quad  0.57178512 \quad 0.33552362]^T,
		\end{align*}
		which is the same as obtained in~\cite{beer1981walsh}.
		We note that the result after the $ 20 $-th iteration is
		\begin{align*}
			x_1&=  [0.11960845 \quad  0.33997421 \quad 0.51220193 \quad  0.62564211]^T, \\
			x_2&=  [0.95686757 \quad  0.80605858 \quad  0.57176313 \quad  0.33575831]^T.
		\end{align*}
		It is known that in this case the analytic solutions are $ \tilde{x}_1(t) = \frac{2t}{t^2 +2} $ and $ \tilde{x}_2(t) = \frac{4-2t^2}{(t^2+2)^2} $. The analytic solutions $ x_1(t) $ and $ x_2(t) $ along with solutions obtained by Walsh functions of order $ N=4 $ and $ N=16 $ are shown in Figure~$ \ref{fig_walsh_approximation_solution_example_one} $.

		\begin{figure}[ht]
			
			%	\resizebox{.9\linewidth}{!}{
				
				\begin{center}
					\begin{subfigure}[t]{.4\textwidth}
						\centering
						\includegraphics[width=\linewidth]{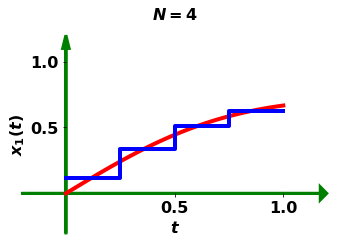}
					\end{subfigure}
					%	\hfill
					\hspace{1cm}
					\begin{subfigure}[t]{.4\textwidth}
						\centering
						\includegraphics[width=\linewidth]{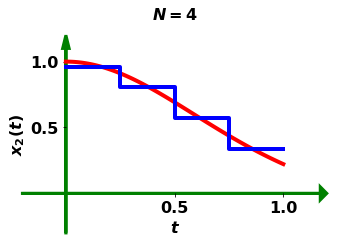}
					\end{subfigure}
					
					\medskip
					
					\begin{subfigure}[t]{.4\textwidth}
						\centering
						\includegraphics[width=\linewidth]{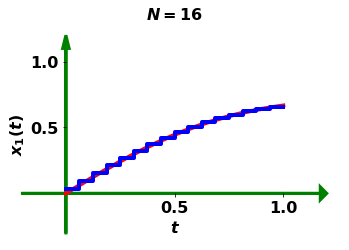}
					\end{subfigure}
					%	\hfill
					\hspace{1cm}
					\begin{subfigure}[t]{.4\textwidth}
						\centering
						\includegraphics[width=\linewidth]{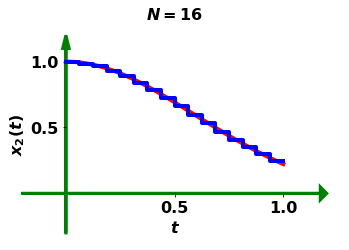}
					\end{subfigure} 
				\end{center} 
				%}
			\caption{Analytic solution in red for $ x_1(t) $ and the solutions obtained by Walsh functions of order $ N=4 $ and $ N=16 $ (in blue), based on the proposed hybrid classical-quantum approach are shown in the left column. The corresponding results for $ x_2(t) $ are shown in the right column.}\label{fig_walsh_approximation_solution_example_one}
		\end{figure}

		Although we demonstrated the application of our proposed hybrid classical-quantum approach for solutions of nonlinear IVPs, the approach could easily be applied to solution of nonlinear Boundary Value
		Problems (BVPs) as well. A well known numerical technique for solution of boundary value problems is the shooting method (covered in many standard textbooks relevant to numerical methods, e.g., \cite{bulirsch2002introduction}, \cite{suli2003introduction}, \cite{burden2015numerical},  \cite{isaacson2012analysis} and \cite{ascher1995numerical}) where a BVP is converted to an IVP and the solution is obtained through an iterative approach. Hence our proposed hybrid classical-quantum approach could be used to solve nonlinear BVPs via the shooting method and the benefits of our proposed approach also extend to the nonlinear BVPs. 

		\subsubsection{Error Analysis}
		
		As we noted earlier the correct analytic solutions are $ \tilde{x}_1(t) = \frac{2t}{t^2 +2} $ and $ \tilde{x}_2(t) = \frac{4-2t^2}{(t^2+2)^2} $ for the problem described in Example 2.
		We note that the accuracy of the solution using the hybrid classical-quantum approach described above depends on the chosen values of $ N $ and $ N_{max} $, where $ N $ and $ N_{max} $ denote the number of basis functions and number of iterations, respectively.	These denote the number of basis functions and number of iterations, respectively
		
		Assume that $ x_1 = [x_{1,0} \quad x_{1,1}  \quad ~\ldots~ \quad x_{1,N-1} ]^T $ and  $ x_2 = [x_{2,0} \quad x_{2,1}  \quad ~\ldots~ \quad x_{2,N-1} ]^T $ are the discrete solutions obtained at the collocation points after $ N_{max} $ iterations. The corresponding functions $ x_1(t) $ and $ x_2(t) $ are obtained as follows,
		
		\begin{equation}\label{eq_def_x1}
			x_k(t) := \begin{cases}
				& x_{k,i} \qquad \frac{i}{N} \leq t < \frac{i+1}{N}, \,  0 \leq i \leq N-1,\\
				&x_{k,N-1} \qquad $ t = 1, $
			\end{cases} 
		\end{equation}  
		for $ k=1 $ and $k= 2 $ respectively. To quantify the accuracy of the numerical solution, we consider an $ L_2 $ error $ \epsilon $ as defined by 
		\begin{equation}\label{eq_error_1}
			\epsilon =  \sqrt{ \int_{0}^{1} \, \left[\sum_{k=1}^{2} \, (\tilde{x}_k(t) - x_k(t))^2  \right]  \, dt }. 
		\end{equation}
		Variation of this error $\epsilon$ with $N$ (for $N = 2^p$, and $p$ = 2, \ldots, 13) is shown in Figure $  \ref{fig:error}$. We observe that the error $\epsilon $ decreases with increasing $N$. A more detailed error analysis (including topics relevant to function approximations using Walsh-Hadamard basis functions, spectral radius and domains of convergence associated with fixed point maps) is outside of the scope of this paper and could be pursued as part of future work.
		\begin{center}
			\begin{figure}
				\centering
				\includegraphics[width=0.6\linewidth]{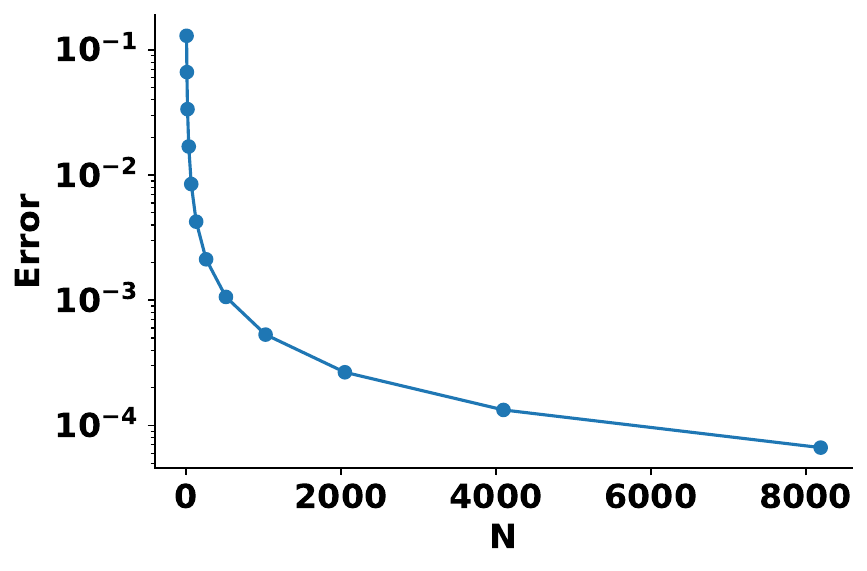}
				\caption{Convergence of $L_2$  error $ \epsilon $  in solution of the system of differential equations given in $ \meqref{eq_ode_one} $ and $ \meqref{eq_ode_two} $ using Algorithm~$ \ref{alg_main} $.}
				\label{fig:error}
			\end{figure}
		\end{center}

		\section{Conclusion}\label{sec:conclusion}
		In this work, we proposed a hybrid classical-quantum approach for solutions of nonlinear ordinary differential equations using Walsh-Hadamard basis functions. 
		The Walsh-Hadamard transform is a key step in many known classical methods for the solution of nonlinear ordinary differential equations using the Walsh-Hadamard basis functions. While Hadamard gates are commonly used in many quantum algorithms and can naturally compute the Walsh-Hadamard transform (under certain conditions), there are challenges involved in extracting useful classical information. These challenges associated with obtaining Walsh-Hadamard transforms of arbitrary vectors are overcome via our proposed hybrid classical-quantum approach (see 
		Algorithm~$ \ref{alg_QWHT} $). This hybrid classical-quantum approach for Walsh-Hadamard transform involves shifting, scaling and measurement operations (along with state preparation and use of quantum Hadamard gates). This hybrid classical-quantum approach for obtaining the Walsh-Hadamard transform for an arbitrary input vector was implemented and successfully tested on the simulation environment on Qiskit (IBM's open source quantum computing platform). The advantage of the proposed hybrid classical-quantum approach for Walsh-Hadamard transform is significantly lower computational complexity ($ (\mathcal{O}(N)) $ operations) in comparison to the classical Fast Walsh-Hadamard transform ($ \mathcal{O}(N \log_2(N)) $ operations). This speedup could be utilized to achieve a superior hybrid classical-quantum approach for the solution of nonlinear ordinary differential equations. The proposed approach was demonstrated using three computational examples relevant to nonlinear differential equations for cases involving one dependent variable (including stiff and non-stiff cases) and also two dependent variables. The results based on the proposed hybrid classical-quantum approach for the solution of nonlinear differential equations were found to be encouraging and they matched the results obtained from corresponding classical approaches (and exact solutions). 
		
		The relation between Walsh-Hadamard functions and the character theory of finite groups is also explored to provide an alternate and perhaps a more conceptually rigorous approach for constructing Walsh-Hadamard functions. The character theory also provides a natural ordering of Walsh-Hadamard basis functions, compatible with unitary transformations associated with quantum Hadamard gates. 
		In contrast to previous works relevant to classical solutions of ordinary differential equations based on Walsh functions, which mostly focused on sequency ordering, this work presents new formulations and results based on the natural ordering of Walsh-Hadamard basis functions.
		
		Future work involves extensions of the proposed hybrid classical-quantum approach to the solution of partial differential equations and quantum machine learning algorithms. Many classical applications (including signal processing and image compression) where classical Walsh-Hadamard transforms are currently used can also benefit from the hybrid classical-quantum approach for computing Walsh-Hadamard transform (a key step in the proposed ODE solution approach).

%		
%		\bibliographystyle{unsrt}
%		\bibliography{Bibliography}

\begin{thebibliography}{10}
			
			\bibitem{butcher2016numerical}
			John~Charles Butcher.
			\newblock {\em Numerical methods for ordinary differential equations}.
			\newblock John Wiley \& Sons, 2016.
			
			\bibitem{suli2010numerical}
			Endre S{\"u}li.
			\newblock Numerical solution of ordinary differential equations.
			\newblock {\em Mathematical Institute, University of Oxford}, 2010.
			
			\bibitem{trefethen2000spectral}
			Lloyd~N Trefethen.
			\newblock {\em Spectral methods in MATLAB}.
			\newblock SIAM, 2000.
			
			\bibitem{boyd2001chebyshev}
			John~P Boyd.
			\newblock {\em Chebyshev and Fourier spectral methods}.
			\newblock Courier Corporation, 2001.
			
			\bibitem{walsh1923closed}
			Joseph~L Walsh.
			\newblock A closed set of normal orthogonal functions.
			\newblock {\em American Journal of Mathematics}, 45(1):5--24, 1923.
			
			\bibitem{beauchamp1975walsh}
			Kenneth~George Beauchamp.
			\newblock Walsh functions and their applications.
			\newblock 1975.
			
			\bibitem{zarowski1985spectral}
			C~Zarowski and Maurice Yunik.
			\newblock Spectral filtering using the fast walsh transform.
			\newblock {\em IEEE transactions on acoustics, speech, and signal processing},
			33(5):1246--1252, 1985.
			
			\bibitem{kuklinski1983fast}
			WS~Kuklinski.
			\newblock Fast walsh transform data-compression algorithm: Ecg applications.
			\newblock {\em Medical and Biological Engineering and Computing},
			21(4):465--472, 1983.
			
			\bibitem{lu2016walsh}
			Yi~Lu and Yvo Desmedt.
			\newblock Walsh transforms and cryptographic applications in bias computing.
			\newblock {\em Cryptography and Communications}, 8(3):435--453, 2016.
			
			\bibitem{beer1981walsh}
			Tom Beer.
			\newblock Walsh transforms.
			\newblock {\em American Journal of Physics}, 49(5):466--472, 1981.
			
			\bibitem{ahner1988walsh}
			Henry~F Ahner.
			\newblock Walsh functions and the solution of nonlinear differential equations.
			\newblock {\em American Journal of Physics}, 56(7):628--633, 1988.
			
			\bibitem{chen1975walsh}
			CF~Chen and CH~Hsiao.
			\newblock A walsh series direct method for solving variational problems.
			\newblock {\em Journal of the Franklin Institute}, 300(4):265--280, 1975.
			
			\bibitem{gnoffo2014global}
			Peter~A Gnoffo.
			\newblock Global series solutions of nonlinear differential equations with
			shocks using walsh functions.
			\newblock {\em Journal of Computational physics}, 258:650--688, 2014.
			
			\bibitem{gnoffo2015unsteady}
			Peter~A Gnoffo.
			\newblock Unsteady solutions of non-linear differential equations using walsh
			function series.
			\newblock In {\em 22nd AIAA Computational Fluid Dynamics Conference}, page
			2756, 2015.
			
			\bibitem{gnoffo2017solutions}
			Peter~A Gnoffo.
			\newblock Solutions of nonlinear differential equations with feature detection
			using fast walsh transforms.
			\newblock {\em Journal of Computational Physics}, 338:620--649, 2017.
			
			\bibitem{nielsen2002quantum}
			Michael~A. Nielsen and Isaac Chuang.
			\newblock {\em Quantum {C}omputation and {Q}uantum {I}nformation}.
			\newblock Cambridge University Press, 2000.
			
			\bibitem{kunz1979equivalence}
			Henry~O. Kunz.
			\newblock On the equivalence between one-dimensional discrete walsh-hadamard
			and multidimensional discrete fourier transforms.
			\newblock {\em IEEE Transactions on Computers}, 28(03):267--268, 1979.
			
			\bibitem{geadah1977natural}
			Youssef~A. Geadah and MJG Corinthios.
			\newblock Natural, dyadic, and sequency order algorithms and processors for the
			walsh-hadamard transform.
			\newblock {\em IEEE Transactions on Computers}, 26(05):435--442, 1977.
			
			\bibitem{childs2020quantum}
			Andrew~M Childs and Jin-Peng Liu.
			\newblock Quantum spectral methods for differential equations.
			\newblock {\em Communications in Mathematical Physics}, 375(2):1427--1457,
			2020.
			
			\bibitem{berry2017quantum}
			Dominic~W Berry, Andrew~M Childs, Aaron Ostrander, and Guoming Wang.
			\newblock Quantum algorithm for linear differential equations with
			exponentially improved dependence on precision.
			\newblock {\em Communications in Mathematical Physics}, 356(3):1057--1081,
			2017.
			
			\bibitem{berry2014high}
			Dominic~W Berry.
			\newblock High-order quantum algorithm for solving linear differential
			equations.
			\newblock {\em Journal of Physics A: Mathematical and Theoretical},
			47(10):105301, 2014.
			
			\bibitem{leyton2008quantum}
			Sarah~K Leyton and Tobias~J Osborne.
			\newblock A quantum algorithm to solve nonlinear differential equations.
			\newblock {\em arXiv preprint arXiv:0812.4423}, 2008.
			
			\bibitem{lloyd2020quantum}
			Seth Lloyd, Giacomo De~Palma, Can Gokler, Bobak Kiani, Zi-Wen Liu, Milad
			Marvian, Felix Tennie, and Tim Palmer.
			\newblock Quantum algorithm for nonlinear differential equations.
			\newblock {\em arXiv preprint arXiv:2011.06571}, 2020.
			
			\bibitem{liu2021efficient}
			Jin-Peng Liu, Herman~{\O}ie Kolden, Hari~K Krovi, Nuno~F Loureiro, Konstantina
			Trivisa, and Andrew~M Childs.
			\newblock Efficient quantum algorithm for dissipative nonlinear differential
			equations.
			\newblock {\em Proceedings of the National Academy of Sciences}, 118(35), 2021.
			
			\bibitem{ascher1995numerical}
			Uri~M Ascher, Robert~MM Mattheij, and Robert~D Russell.
			\newblock {\em Numerical solution of boundary value problems for ordinary
				differential equations}.
			\newblock SIAM, 1995.
			
			\bibitem{bulirsch2002introduction}
			Roland Bulirsch, Josef Stoer, and J~Stoer.
			\newblock {\em Introduction to numerical analysis}, volume~3.
			\newblock Springer, 2002.
			
			\bibitem{suli2003introduction}
			Endre S{\"u}li and David~F Mayers.
			\newblock {\em An introduction to numerical analysis}.
			\newblock Cambridge university press, 2003.
			
			\bibitem{burden2015numerical}
			Richard~L Burden, J~Douglas Faires, and Annette~M Burden.
			\newblock {\em Numerical analysis}.
			\newblock Cengage learning, 2015.
			
			\bibitem{isaacson2012analysis}
			Eugene Isaacson and Herbert~Bishop Keller.
			\newblock {\em Analysis of numerical methods}.
			\newblock Courier Corporation, 2012.
			
			\bibitem{serre1977linear}
			Jean-Pierre Serre.
			\newblock {\em Linear representations of finite groups}, volume~42.
			\newblock Springer, 1977.
			
			\bibitem{fulton2013representation}
			William Fulton and Joe Harris.
			\newblock {\em Representation theory: a first course}, volume 129.
			\newblock Springer Science \& Business Media, 2013.
			
			\bibitem{steinberg2012representation}
			Benjamin Steinberg.
			\newblock {\em Representation theory of finite groups: an introductory
				approach}.
			\newblock Springer, 2012.
			
			\bibitem{terras1999fourier}
			Audrey Terras.
			\newblock {\em Fourier analysis on finite groups and applications}.
			\newblock Number~43. Cambridge University Press, 1999.
			
			\bibitem{harvey2021integer}
			David Harvey and Joris Van Der~Hoeven.
			\newblock Integer multiplication in time o (n log n).
			\newblock {\em Annals of Mathematics}, 193(2):563--617, 2021.
			
			\bibitem{scheffel2012spectral}
			Jan Scheffel.
			\newblock A spectral method in time for initial-value problems.
			\newblock {\em American Journal of Computational Mathematics}, 2(3):173--193,
			2012.
			
		\end{thebibliography}

	\end{document}